\documentclass[conference,letterpaper]{IEEEtran}

\usepackage[utf8]{inputenc} 
\usepackage[T1]{fontenc}
\usepackage{url}
\usepackage{ifthen}
\usepackage{cite}
\usepackage[cmex10]{amsmath}
\usepackage{hyperref}
\usepackage{thmtools, thm-restate}
\usepackage{comment}

\usepackage[usenames,dvipsnames]{xcolor}
\usepackage{dsfont}
\usepackage{amsbsy}
\usepackage{amssymb}
\usepackage{amsmath}%
\usepackage{amsfonts}%
\usepackage{amsthm}
\usepackage{breqn}
\usepackage{graphicx}
\usepackage{colonequals}
\usepackage{psfrag}
\usepackage{subfigure}
\usepackage{mathrsfs}
\usepackage{mdframed}
\usepackage{enumerate}
\usepackage{enumitem}
\usepackage{braket}

\allowdisplaybreaks

\newtheorem{thm}{Theorem}
\newtheorem{lem}{Lemma}

\newtheorem{prop}{Proposition}

\makeatletter
\newtheorem*{rep@theorem}{\rep@title}
\newcommand{\newreptheorem}[2]{%
\newenvironment{rep#1}[1]{%
 \def\rep@title{#2 \ref{##1}}%
 \begin{rep@theorem}}%
 {\end{rep@theorem}}}
\makeatother
\newreptheorem{conj}{Conjecture}

\theoremstyle{definition}

\newtheorem{defn}{Definition}

\newtheorem{remark}{Remark}

\newcommand{\calX}{\mathcal{X}}

\renewcommand{\tilde}{\widetilde}

\newcommand{\defined}{\triangleq}

\newcommand{\TV}{\mathsf{TV}}

\definecolor{light-gray}{gray}{.90}

\newmdenv[%
  backgroundcolor=red, %
  linecolor=black,
  linewidth =1pt,%
  skipabove = 10pt,%
  skipbelow = 10pt
]{TODO}

\usepackage{microtype}
\definecolor{quotemark}{gray}{0.7}
\makeatletter
\def\fquote{%
    \@ifnextchar[{\fquote@i}{\fquote@i[]}
           }%
\def\fquote@i[#1]{%
    \def\tempa{#1}%
    \@ifnextchar[{\fquote@ii}{\fquote@ii[]}
                 }%
\def\fquote@ii[#1]{%
    \def\tempb{#1}%
    \@ifnextchar[{\fquote@iii}{\fquote@iii[]}
                      }%
\def\fquote@iii[#1]{%
    \def\tempc{#1}%
    \vspace{1em}%
    \noindent%
    \begin{list}{}{%
         \setlength{\leftmargin}{0.1\textwidth}%
         \setlength{\rightmargin}{0.1\textwidth}%
                  }%
         \item[]%
         \begin{picture}(0,0)%
         \put(-15,-5){\makebox(0,0){\scalebox{3}{\textcolor{quotemark}{``}}}}%
         \end{picture}%
         \begingroup\itshape}%
 \def\endfquote{%
 \endgroup\par%
 \makebox[0pt][l]{%
 \hspace{0.8\textwidth}%
 \begin{picture}(0,0)(0,0)%
 \put(15,15){\makebox(0,0){%
 \scalebox{3}{\color{quotemark}''}}}%
 \end{picture}}%
 \ifx\tempa\empty%
 \else%
    \ifx\tempc\empty%
       \hfill\rule{100pt}{0.5pt}\\\mbox{}\hfill\tempa,\ \emph{\tempb}%
   \else%
       \hfill\rule{100pt}{0.5pt}\\\mbox{}\hfill\tempa,\ \emph{\tempb},\ \tempc%
   \fi\fi\par%
   \vspace{0.5em}%
 \end{list}%
 }%
 \makeatother

\newcommand{\KL}{D_{\mathsf{KL}}}

\DeclareMathOperator{\arcsinh}{arcsinh}

\newcommand{\blfootnote}[1]{%
  \begingroup
  \renewcommand{\thefootnote}{}
  \footnotetext{#1}%
  \addtocounter{footnote}{-1}
  \endgroup
}

\begin{document}
\title{Soft Best-of-$n$ Sampling for Model Alignment
\thanks{This work is supported by the National Science Foundation under grants CIF 2312667, FAI 2040880, and CIF 2231707 and the NSF-GRFP under grant DGE-2140743. We also thank Ahmad Beirami for insightful discussions about this project.}}


\author{%
  \IEEEauthorblockN{                 Claudio Mayrink Verdun\textsuperscript{*}\IEEEauthorrefmark{2}, Alex Oesterling\textsuperscript{*}\IEEEauthorrefmark{2},
                   Himabindu Lakkaraju\IEEEauthorrefmark{3},
                    Flavio P. Calmon\IEEEauthorrefmark{2}}
  \IEEEauthorblockA{\IEEEauthorrefmark{2}%
                   Harvard University, Allston, MA, USA,
                    \{claudioverdun,flavio\}@seas.harvard.edu, aoesterling@g.harvard.edu}
  \IEEEauthorblockA{\IEEEauthorrefmark{3}%
                   Harvard Business School, Allston, MA, USA,
                    hlakkaraju@hbs.edu}
}

\maketitle
\blfootnote{\textsuperscript{*}Equal contribution.} 
\blfootnote{This work is supported by the National Science Foundation under grants CIF 2312667, FAI 2040880, and CIF 2231707 and the NSF-GRFP under grant DGE-2140743. We also thank Ahmad Beirami for insightful discussions about this project.}
\setcounter{footnote}{0}

\begin{abstract}
   Best-of-$n$ (BoN) sampling is a practical approach for aligning language model outputs with human preferences without expensive fine-tuning. BoN sampling is performed by generating $n$ responses to a prompt and then selecting the sample that maximizes a reward function. BoN yields high reward values in practice at a distortion cost, as measured by the KL-divergence between the sampled and original distribution. This distortion is coarsely controlled by varying the number of samples: larger $n$ yields a higher reward at a higher distortion cost.

We introduce Soft Best-of-$n$ sampling, a generalization of BoN that allows for smooth interpolation between the original distribution and reward-maximizing distribution through a temperature parameter $\lambda$. We establish theoretical guarantees showing that Soft Best-of-$n$ sampling converges sharply to the optimal tilted distribution at a rate of $O(1/n)$ in KL and the expected (relative) reward. For sequences of discrete outputs, we analyze an additive reward model that reveals the fundamental limitations of blockwise sampling.
\end{abstract}

\section{Introduction}
\label{sec:intro}

Large Language Models (LLMs) are a type of generative machine learning model that, given a string of input tokens, output a distribution $P$ that approximates the likelihood of the next tokens in a sequence. 
By learning from large corpora of text data, LLMs have demonstrated remarkable capabilities in natural language understanding, reasoning, and generation across a range of domains. However, training models for accurate next token prediction does not explicitly optimize for \emph{alignment} between the model's outputs and human preferences \cite{hadfield2016cooperative, leike2018scalable}. As a result, the behavior of LLMs may not always conform to the intentions or expectations of human users \cite{bender2021dangers, bommasani2021opportunities}.

Given a reference distribution  $P$ over a finite alphabet $\mathcal{X}$ and a reward function $r: \mathcal{X} \rightarrow \mathbb{R}$, a popular approach to address the \emph{alignment problem} is via an optimization formulation that is essentially equivalent to the well-known information projection optimization \cite{csiszar2003information, csiszar1975divergence}. The goal is to find a new distribution $P^*$ that is similar to $P$ while maximizing the expected reward:
\begin{equation}
\label{eq:KL-optimization}
\begin{aligned}
    \max_{P^*\in\Delta_{\calX}}~\mathbb{E}_{P^*}\left[r(X) \right]\
    \mbox{subject to } \KL(P^*\|P) \leq \epsilon.
\end{aligned}
\end{equation} 
In LLMs, $P$ is the distribution of whole responses given the prompt (often referred to as ``policy'' in the alignment literature) and $r$ is a proxy model for the human preferences that we seek to align $P$ with. The function $r$ itself is usually another language model trained to score or rate responses \cite{ouyang2022training}. In theory, the solution to this problem is a tilted distribution \cite{csiszar1975divergence}, which adjusts $P$ to place more weight on high-reward responses. The optimal distribution is given by:
\begin{equation}
    \label{eq:Pstar-def}
    P^*_\lambda (x) = \frac{P(x)e^{r(x)/\lambda}}{\mathbb{E}_P[e^{r(X)/\lambda}]}.
\end{equation}
where $\lambda>0$ controls the trade-off between reward maximization and the divergence from the reference distribution.
Variations of information projection and tilted distributions are common in information theory\cite{chentsov1968nonsymmetrical,csiszar1995generalized, csiszar2003information,alghamdi2020model} and have appeared across various fields, including Monte Carlo simulation, large deviations theory, and econometrics \cite{nakayama1992efficient, hofert2011sampling,fuh2024efficient, imbens1998information,kitamura1997information,sadowsky1990large,siegmund1976importance}; see also \cite{asmussen2007stochastic}. 

In generative models, we do not have an analytical expression for $P$ nor $r(X)$, and thus cannot compute $P^*$ directly; instead, we can only sample from $P$ and evaluate $r(X)$ for individual samples. When limited to only samples from $P$, there are generally two strategies for alignment: model fine-tuning using reinforcement learning to maximize reward with a KL penalty \cite{christiano2017deep, stiennon2020learning, rafailov2024direct}, and inference-time methods such as Best-of-$n$ (BoN) which samples $n$ generations from $P$ and then picks the sequence with the highest reward. While fine-tuning methods require human-annotated preference data and additional GPU-intensive training, BoN has emerged as a simple yet effective approach to alignment. 

Empirically, BoN has been shown to be competitive with fine-tuning approaches in terms of KL-reward tradeoffs \cite{gao2023scaling, mudgalcontrolled}, and theoretically, it has been shown to be equivalent to the fine-tuned solution asymptotically \cite{yang2024asymptotics} while \cite{mroueh2024information} bounds the sub-optimality of BoN in the non-asymptotic regime. However, BoN always selects the reward-maximizing generation, limiting our ability to navigate and control the KL-reward tradeoff.

We introduce \emph{Soft Best-of-$n$} sampling, a generalization of BoN sampling that provides non-asymptotic guarantees for approximating the optimal tilted distribution $P_\lambda^*$ with a finite number of samples $n$. While BoN can be viewed as a special case of Soft Best-of-$n$, the latter offers a more flexible approach by incorporating a temperature parameter $\lambda$. This allows for a smooth interpolation between the reference distribution $P$ and the reward-maximizing distribution $P_{\lambda^*}$, enabling finer control over the KL-reward tradeoff.

In this paper, we focus on deriving theoretical guarantees on the performance of Soft Best-of-$n$ sampling relative to the optimal distribution $P_{\lambda}^*$. Our main contributions are:
\begin{itemize}
    \item We establish sharp bounds on the convergence rate of  Soft Best-of-$n$ sampling to the optimal tilted distribution $P_\lambda^*$, showing that the KL-divergence between the two and the relative difference in reward both decrease at a rate of $O(1/n)$ (Theorems \ref{thm:KL-bound} and \ref{thm:relative_bound}) with explicit constants depending on the temperature parameter. We also provide a converse result showing that this rate cannot be improved in general (Theorem \ref{thm:KL-lowerbound}).
    
    \item We analyze the trade-off between reward maximization and approximating $P_\lambda^*$ and demonstrate that, unlike BoN, Soft Best-of-$n$ sampling can provably span the optimal KL-reward Pareto frontier. We show that, as long as $\lambda$ is not too small (at least $O(1/\log(n\epsilon)$), the Soft Best-of-$n$ sampling is $\epsilon$-close in  KL-divergence to $P_\lambda^*$ (Corollary \ref{cor:lambda_lower_bound}).
    
    \item We study an additive reward model that provides insights into the limitations of using Best-of-$n$ with blockwise sampling and motivates efficient sampling alternatives.
\end{itemize}

\noindent \textbf{Prior Results.} While multiple approaches to fine-tuning with a KL constraint \cite{schulman2017proximal, ouyang2022training, rafailov2024direct, gui2024bonbon} and efficiency improvements on BoN \cite{jinnai2024regularized, ichihara2025evaluation,amini2024variational, qiu2024treebon, zhang2024accelerating} have been recently proposed, in this work we consider the original BoN procedure. Prior works claim that the divergence of BoN sampling from the reference distribution scaled exactly as $\log(n)-(n-1)/n$ \cite{stiennon2020learning, nakano2021webgpt, hilton2022measuringgoodhart, gao2023scaling}. Remarkably, recent results by \cite[Thm. 1]{beirami2024theoretical} and \cite[Thm. 1]{mroueh2024information} prove that this is in fact an upper bound on the actual KL-divergence. Empirically, \cite{gao2023scaling} and \cite{hilton2022measuringgoodhart} found that the KL-reward tradeoff for both RL and BoN methods scaled as $\sqrt{\KL}$ between the optimized and base distribution, and \cite{mroueh2024information} demonstrated that this relationship is an information-theoretic upper bound. The above results characterize KL-reward relationships and the drift of BoN sampling from the reference distribution $P$. However, they do not provide any insights into the optimality of BoN, i.e., its KL divergence from the exponential tilting solution $P_\lambda^*$ in \eqref{eq:Pstar-def}. \cite{yang2024asymptotics} show that BoN is asymptotically equivalent to the KL-constrained reinforcement learning solution and \cite{mroueh2024information} establish non-asymptotic guarantees between the two. We provide tight non-asymptotic bounds on the optimality of Soft Best-of-$n$, which, as a generalization of BoN and the base model, allow us to draw connections between the convergence of each sampling approach.

\section{Preliminaries}
We use uppercase letters (e.g., $X$) to denote random variables and lowercase letters (e.g., $x$) for their realizations. We use superscripts to denote sequence length (e.g., $x^m = (x_1, \ldots, x_m)$) and subscripts to denote sequence elements (e.g., $x_i$ is the $i$-the element of sequence $x^m$). For a finite set $\mathcal{X}$, we denote by $\Delta_{\mathcal{X}}$ the probability simplex over $\mathcal{X}$, i.e., $\Delta_{\mathcal{X}} = \left\{P : \mathcal{X} \to [0,1] \;\middle|\; \sum_{x \in \mathcal{X}} P(x) = 1\right\}$. We define the \emph{coefficient of variation} of a random variable $X$ with finite mean and variance as $\mathsf{CV}(X) \triangleq \sqrt{\frac{\mathsf{Var}(X)}{\mathbb{E}[X]^2}}$.

Let $P$ be a distribution over a finite alphabet $\mathcal{X}$ and $r : \mathcal{X} \to \mathbb{R}$ be a $P$-measurable reward function. Recall that the goal of alignment is to find a new distribution $P^*$ that is similar to $P$ but also maximizes the expected reward \eqref{eq:KL-optimization}, while being limited to only samples from $P$.

A popular approach to this problem is \emph{Best-of-$n$ (BoN) sampling} \cite{stiennon2020learning}, which draws $n$ samples from $P$ and returns the one with the highest reward.
\begin{defn}\label{defn:best-of-n}
For a fixed integer $n \geq 1$, the \emph{Best-of-$n$} sampling strategy consists of:
\begin{enumerate}
    \item Draw $X_1, \ldots, X_n$ i.i.d.\ from $P$;
    \item Compute $R_1 = r(X_1), \ldots, R_n = r(X_n)$;
    \item Return $Y = X_Z$, where $Z = \arg\max_{j=1,\ldots,n} R_j$.
\end{enumerate}
\end{defn}

While effective, BoN sampling can be viewed as an extreme case of a more general framework. We propose \emph{Soft Best-of-$n$ sampling}, which offers a smoother approach to trading off between reward maximization and distribution similarity.
\begin{defn}
\label{defn:n-tilt}
For a fixed integer $n \geq 1$ and $\lambda > 0$, the \emph{Soft Best-of-$n$ sampling} strategy consists of:
\begin{enumerate}
    \item Draw $X_1, \ldots, X_n$ i.i.d.\ from $P$;
    \item Draw $Z$ from $\{1, \ldots, n\}$ with distribution
   \begin{equation} \label{eq:n_tilt_def}
        \Pr(Z=i) = \frac{e^{r(X_i)/\lambda}}{\sum_{j=1}^n e^{r(X_j)/\lambda}};
    \end{equation}
    \item Return $Y = X_Z$.
\end{enumerate}
The distribution of $Y$ produced by the Soft Best-of-$n$ sampling for a fixed $\lambda>0$ is denoted by $P_{n,\lambda}$. 
\end{defn}
Here, $\lambda > 0$ is a temperature parameter that controls the trade-off between reward maximization and similarity to the original distribution $P$. As $\lambda \to 0$, Soft Best-of-$n$ sampling converges to BoN sampling, making the latter a special case of our more general framework. This relationship can also be understood through the lens of the ``softmax'' nature of the sampling probabilities: as $\lambda \rightarrow0$, the softmax becomes increasingly sharp, converging to the argmax function that characterizes BoN sampling. Conversely, as $\lambda \to \infty$, the sampling becomes uniform among the $n$ candidates, effectively recovering the original distribution $P$. Also, observe that Soft Best-of-$n$ sampling only requires i.i.d. samples from $P$ and computation of $r(X_i)$ for the selected samples. 
The distribution induced by Soft Best-of-$n$ sampling approximates the optimal solution \eqref{eq:Pstar-def} to problem \eqref{eq:KL-optimization}, i.e., for any fixed $\lambda$, as $n \rightarrow \infty$, it converges to $P^*_{\lambda}$.

\section{Guarantees for Soft Best-of-$n$ Sampling}
\label{sec:theory_KL}
This section aims to characterize the convergence rate (in KL-divergence) of $P_{n,\lambda} \to P^*_\lambda$ with $n$. Our main result shows that the KL divergence between these distributions decreases at a rate of $O(1/n)$. Naturally, convergence will depend on the value of $\lambda$: if $\lambda \ll 1$, then $n$ has to be large for $P_{n,\lambda} \approx P^*_\lambda$. To see why, as $\lambda \to 0$, $P_\lambda^*$ converges to the distribution where the output $x^*$ corresponding to the highest reward $r(x^*)$ has probability 1. If we aim to approximate this optimal distribution via $P_{n,\lambda}$, then $n$ has to be sufficiently large such that the output $x^*$ will be sampled with high probability, so $n> 1/P(x^*)$.
\begin{restatable}{lem}{LemSoftBon} \label{lem:p-char}
Let  $n \geq 1$, $\lambda > 0$, $r_\lambda(x)= r(x)/\lambda$, and $X_1,..., X_n\overset{i.i.d.}{\sim}P$. For any $x \in \mathcal{X}$, the distribution $P_{n,\lambda}$ of the sample resulting from Soft Best-of-$n$ sampling satisfies
\begin{align}
    P_{n,\lambda}(x) &= P(x)e^{r_\lambda(x)} \times \mathbb{E}\left[\frac{1}{\frac{1}{n}\left(e^{r_\lambda(x)} + \sum_{i=1}^{n-1} e^{r_\lambda(X_i)} \right)} \right] \label{eq:prob-exact}\\
    &\geq P(x)e^{r_\lambda(x)}\times\left(\frac{1}{n}e^{r_\lambda(x)}+\frac{n-1}{n}\mathbb{E}_P\left[e^{r_\lambda(X)} \right]\right)^{-1}. \label{eq:prob-lb}
\end{align}
\end{restatable}
\begin{proof}[Proof sketch of Lemma \ref{lem:p-char}]
The proof exploits symmetry in the definition of $P_{n,\lambda}(x)$. This symmetry allows to reduce the analysis to the case where $Z=1$ as the chosen sample in Defn. \ref{defn:n-tilt} (up to a scaling factor of $n$). The proof then follows by computing  $\Pr(Z = 1|X_1 = x)\times P_{\lambda}^*(x),$ yielding \eqref{eq:prob-exact}, then applying Jensen's inequality, resulting in the lower bound \eqref{eq:prob-lb}.
\end{proof}
\vspace{-10pt}
\begin{restatable}{thm}{ThmKLUpper}\label{thm:KL-bound}
For $\lambda > 0$ and integer $n \geq 1$,
\begin{equation}
\KL\left(P_\lambda^* \| P_{n,\lambda}\right) \leq \log\left(1 + \frac{1}{n}\mathrm{CV}\left(e^{r_\lambda(X)}\right)^2\right).\label{eq:KL-bound}
\end{equation}
Consequently, $\KL\left(P_\lambda^* \| P_{n,\lambda}\right) = O(1/n)$. Moreover, if $0 \leq r(x) \leq 1$ for all $x \in \mathcal{X}$, then 
\begin{equation}
\label{eq:sinh_bound}
  \KL\left(P^*_\lambda \| P_{n,\lambda}\right) \leq \frac{1}{n}\sinh\left(\frac{1}{2\lambda}\right)^2.  
\end{equation}
\end{restatable}
\begin{proof}[Proof sketch of Theorem \ref{thm:KL-bound}] The key insight is to leverage the explicit form of $P_{n,\lambda}$ derived in Lemma \ref{lem:p-char}. We can thus bound the KL divergence
    \begin{align*}
         &\KL\left(P^*_\lambda \| P_{n,\lambda}\right) = \sum_x P^*_\lambda(x) \log \left(\frac{ P^*_\lambda(x) }{ P_{n,\lambda}(x)} \right) \\
         &\leq \sum_x \frac{P(x)e^{r_\lambda(x)}}{\mathbb{E}_P[e^{r_\lambda(X)}]} \log\left( \frac{1/\mathbb{E}_P[e^{r_\lambda(X)}]}{1/\left(\frac{1}{n}e^{r_\lambda(x)}+\frac{n-1}{n}\mathbb{E}_P\left[e^{r_\lambda(X)} \right]\right)} \right)
    \end{align*}
Using Jensen's inequality, we obtain:
\[\KL(P^*_\lambda \| P_{n,\lambda}) \leq \log\left(1 + \frac{1}{n}\text{CV}\left(e^{r_\lambda(X)}\right)^2\right)\]
where $\text{CV}$ denotes the coefficient of variation. When $0 \leq r(x) \leq 1$, we can further bound the coefficient of variation using properties of exponential functions and the Bhatia-Davis inequality \cite{bhatia2000better} to obtain the bound $\mathsf{CV}\left(e^{r_\lambda(X)}\right)^2\leq  \frac{(e^{1/\lambda}-1)^2}{4e^{1/\lambda}}$, leading to \eqref{eq:sinh_bound}.
\end{proof}

The bounds in Theorem \ref{thm:KL-bound} can be inverted to derive a sufficient condition  $\lambda$ to guarantee that $P_{n,\lambda}$ is ``close'' (in KL-divergence) to the optimal tilted distribution $P_\lambda^*$.
\begin{restatable}{cor}{CorLambda}\label{cor:lambda_lower_bound}
For $0 \leq r(x) \leq 1$ and fixed $n \geq 1$ and $\epsilon > 0$, if $ \lambda \geq \frac{1}{\log(1 + 4n\epsilon)},$ then $\KL\left(P^*_\lambda \| P_{n,\lambda}\right) \leq \epsilon$.
\end{restatable}

\begin{figure}[t]
    \centering
    \includegraphics[width=0.78\columnwidth]{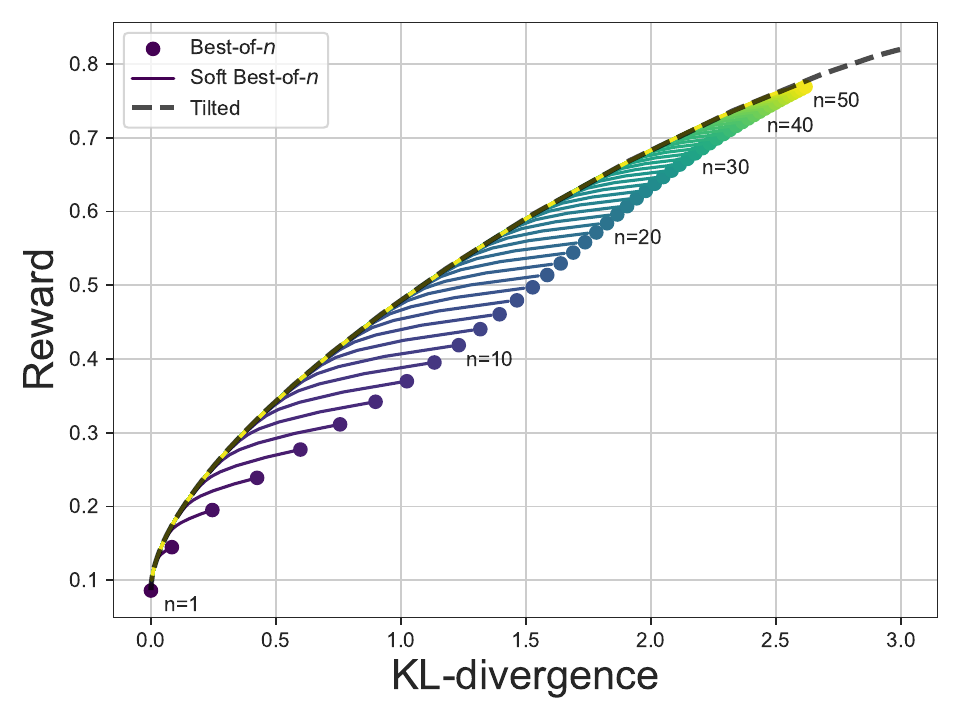}
    \caption{Soft Best-of-$n$ and BoN sampling, compared with the optimal Pareto frontier of exponential tilting for KL-reward tradeoffs. Soft Best-of-$n$ sampling generalizes BoN and allows for control between KL and reward, allowing us to achieve near-optimal performance for all $n$. Alphabet of size 3 with distribution $[0.75, 0.2, 0.05]$ and reward $[0.016, 0.164, 0.820]$.}
    \label{fig:pareto}
    \vspace{-15pt}
\end{figure}

Importantly, Corollary \ref{cor:lambda_lower_bound} demonstrates that Theorem \ref{thm:KL-bound} does not provide a meaningful bound on $\KL\left(P^*_\lambda \| P_{n,\lambda}\right)$ when $\lambda$ is small. Recall that Soft Best-of-$n$ sampling converges to BoN sampling as $\lambda \rightarrow 0$, and it is precisely in this regime where our results do not guarantee near-optimality. For a fixed $n$, guaranteeing that BoN sampling approximates $P_\lambda^*$ is not possible in general. As an example, consider Figure \ref{fig:pareto}. Here, BoN sampling is bounded away from the optimal tilted distribution except when $n$ is large ($n\approx 50$ in the plot) and, consequently, is far from the Pareto-optimal frontier. In contrast,  by simultaneously controlling $\lambda$ and $n$, Soft Best-of-$n$ sampling is able to approximate the optional KL-reward Pareto frontier as long as $\lambda$ is sufficiently large (cf. Corollary \ref{cor:lambda_lower_bound}).

The upper bound in Theorem \ref{thm:KL-bound} is, in fact, sharp, as demonstrated through a binary example. Specifically, consider $\mathcal{X} = \{0,1\}$ with $P(1) = p$ and a reward function $r(x) = \mathds{1}_{\{x=1\}}$. We can explicitly compute $\KL\left(P_\lambda^* \| P_{n,\lambda}\right) =  \frac{1}{n} \frac{p(1-p)(e^{1/\lambda}-1)^2}{(pe^{1/\lambda}+(1-p))^2}$, matching the $O(1/n)$ rate of Theorem \ref{thm:KL-bound}\footnote{\vspace{-10pt}See Appendix \ref{sec:sequence_upperbound} further analysis of this construction.}. 

Next, we show that the $O(1/n)$ decay rate in KL-divergence is tight. To show this result, we focus on a lower bound for $\TV(P^*_\lambda \| P_{n,\lambda})$. The next theorem establishes a lower bound of $O(1/\sqrt{n})$ on the convergence rate of $\TV(P^*_\lambda \| P_{n,\lambda})$ which, via Pinsker's inequality, gives an $O(1/n)$ lower bound on $\KL(P^*_\lambda \| P_{n,\lambda}),$ though the TV convergence rate is of independent interest. 

\begin{restatable}{thm}{ThmKLLower}\label{thm:KL-lowerbound}
For $\lambda > 0$ and integer $n \geq 1$, the following bounds hold on the TV distance:
\begin{equation}
\TV(P^*_\lambda \| P_{n,\lambda}) \geq  \frac{C'}{\sqrt{n-1}},
\end{equation}
where $C'$ is a constant depending on $\lambda$, the reward function $r$, and the second and third moments of $e^{r(X)/\lambda}$. 
\end{restatable}

\begin{proof}[Proof sketch of Theorem \ref{thm:KL-lowerbound}] 
We bound the total variation distance between $P^*_\lambda$ and $P_{n,\lambda}$ by considering large enough $n$ and applying the Berry-Esseen theorem \cite{berry1941accuracy, esseen1942liapounoff}. Let $\mu = \mathbb{E}_P[e^{r(X)}].$ By carefully examining the selection probability in step 2 of Definition 2, we can show that for any $\delta < 1-\frac{1}{\mu}$:
\begin{align*}
    \TV(P^*_\lambda \|P_{n,\lambda}) &\geq \frac{1}{2}\sum_x \frac{P(x)e^{r(x)}}\mu |1 - L_n(x)| \\
    &\geq \frac{\delta}{4} - \frac{C'}{\sqrt{n-1}} - O\left(\frac{1}{n^{3/2}}\right)
\end{align*}
where $L_n(x)$ captures the sampling dynamics and $C'$ is a constant depending on the reward properties. The result follows by optimizing the choice of $\delta$. By applying Pinsker's inequality \cite{cover1999elements}, we derive the bound $\KL(P^*_\lambda \|P_{n,\lambda}) \geq \Omega (1/n)$.
\end{proof}

Next, we derive a bound on the convergence of the expected reward under $P_{n,\lambda}$ relative to $P_\lambda^*$, assuming that $0\leq r(x)\leq 1$. One strategy is to leverage the result of Theorem \ref{thm:KL-bound} and directly apply (yet again) Pinsker's inequality, which would yield an upper-bound on total variation and, consequently, on the expectation of any bounded reward function $r(X)$. However, this route renders a convergence rate of $O(1/\sqrt{n})$ in reward. The next bound proves a \emph{relative} convergence rate of $O(1/n)$ in expected reward, albeit at the cost of additional constants. We focus on bounding only one side of the reward, i.e., how much greater the optimal reward under $P_\lambda^*$ is relative to $ P_{n,\lambda}$.

\begin{restatable}{thm}{ThmRewardBound}\label{thm:relative_bound}
    Let $0\leq r(x) \leq 1$ and $M(\lambda)\defined (e^{1/\lambda}/\mathbb{E}_P[e^{r_\lambda(X)}])-1.$ Then 
    \begin{equation}
        \frac{\mathbb{E}_{P_\lambda^*}[r_\lambda(X)] - \mathbb{E}_{P_{n,\lambda}}[r_\lambda(X)]}{\mathbb{E}_{P_\lambda^*}[r_\lambda(X)]} \leq \frac{M(\lambda)}{M(\lambda)+n}.
    \end{equation}
    In particular, the reward achieved by Soft Best-of-$n$ sampling converges to that of the tilted distribution with relative error $O(1/n).$
\end{restatable}

\begin{proof}[Proof of Theorem \ref{thm:relative_bound}]
 For $P_\lambda^*$ given in \eqref{eq:Pstar-def},  $ P_{n,\lambda}$ given in \eqref{eq:n_tilt_def}, $r(x)\geq 0$, and applying \eqref{eq:prob-lb}:
     \begin{align*}
         \mathbb{E}_{P_\lambda^*}[r_\lambda(X)] - \mathbb{E}_{P_{n,\lambda}}[r_\lambda(X)] \\
         \leq \sum_x \frac{P(x)e^{r_\lambda}(x)r_\lambda(x)}{\mathbb{E}_{P}[e^{r_\lambda(X)}]}\left(1 - \frac{\mathbb{E}_{P}[e^{r_\lambda(X)}] }{\frac{1}{n}e^{r_\lambda(x)}+ \frac{n-1}{n} E_P[e^{r_\lambda(X)}]} \right) 
     \end{align*}

\noindent Since $0 \leq r(x) \leq 1$, we have $e^{r_\lambda(x)} \leq e^{1/\lambda}$ for all $x$ and:
\begin{equation}
\frac{e^{r_\lambda(x)}}{\mathbb{E}_P[e^{r_\lambda(X)}]} \leq \frac{e^{1/\lambda}}{\mathbb{E}_P[e^{r_\lambda(X)}]} = M(\lambda) + 1
\end{equation}

The relative difference in expected rewards can be bounded:
\[
\frac{\mathbb{E}_{P^*_\lambda}[r_\lambda(X)] - \mathbb{E}_{P_{n,\lambda}}[r_\lambda(X)]}{\mathbb{E}_{P^*_\lambda}[r_\lambda(X)]} \leq \frac{M(\lambda)}{M(\lambda) + n}
\]

As $\lambda$ increases, both the numerator $e^{1/\lambda}$ and $\mathbb{E}_P[e^{r(X)/\lambda}]$ in $M(\lambda)$ decrease and approach 1, causing $M(\lambda)\rightarrow0$. This leads to faster convergence in the relative error bound, albeit at the cost of reduced reward discrimination between sequences.
\end{proof}

This improved convergence comes at a cost: larger values of $\lambda$ also reduce the effective reward difference between sequences (since we are dividing rewards by $\lambda$), making the sampling distribution closer to the original distribution $P$. This creates a fundamental trade-off: smaller $\lambda$ values allow for more aggressive optimization of the reward but require more samples for convergence (as $M(\lambda)$ is large), while larger $\lambda$ values provide faster convergence but result in more conservative optimization (as $M(\lambda)$ is small). 
Our results highlight that there is a delicate balance between $\lambda$ and $n$ in order to ensure that $P_{n,\lambda}$ approximates $P_\lambda^*$. Such a balance is not feasible in general when restricted to BoN sampling.

\section{The Additive Reward Model}
We analyze next a simplified setting in which the reward of a sequence of outcomes is additive (or ``separable'' in rate-distortion parlance). Consider an i.i.d. source $P_X^{\otimes{m}}$ over $\calX^m$ and a reward function $r:\calX\to [0,1]$. The reward of a sequence $x^m=(x_1,\dots,x_m)$ is given by $r(x^m) \triangleq \frac{1}{m}\sum_{i=1}^m r(x_i)$.
Of course, for an LLM, the source will absolutely not be i.i.d., and the reward will almost certainly not be additive. Nevertheless, this model will help us to elucidate the limitations of  best-of-$n$ sampling and  motivate controlled decoding strategies. We consider two different sampling strategies to maximize reward: \emph{symbolwise}, where each symbol is sampled $n$ times and selected based on their respective reward, and \emph{blockwise}, where instead each block of $m$ symbols is sampled $n$ times. 

\noindent \textbf{A. Tilting Symbols vs. Tilting Blocks.} When both $P$ and $r$ are known, the same KL-divergence vs. reward trade-off curve is spanned by either solving \eqref{eq:KL-optimization} at the individual symbol level or for a sequence of $m$ symbols. To see why this is the case, consider blockwise counterpart of \eqref{eq:KL-optimization}, given by the optimization
\begin{equation}
\label{eq:KL-optimization-iid}
\begin{aligned}
\max_{P^{*m}\in\Delta_{\calX^m}}~\mathbb{E}_{P^{*m}}\left[r(X) \right]\
    \mbox{subject to } \KL(P^{*m}\|P^{\otimes m}) \leq \epsilon.
\end{aligned}
\end{equation} 
Naturally, the optimal solution $P^{*m}$ is of the form
\begin{align}
    P^{*m}_{\lambda} (x^m) = \frac{P^{\otimes m}(x^m)e^{r(x^m)/\lambda}}{\mathbb{E}_{P^{\otimes m}}[e^{r(X^m)/\lambda}]}
    = \prod_{i=1}^m  \frac{P(x_i)e^{r(x_i)/m\lambda}}{\mathbb{E}_{P}[e^{r(X)/m\lambda}]}
\end{align}
and $P^{*m}_{\lambda}=P^{*\otimes m}_{m\lambda}$.  

The simple derivation above shows that to maximize reward subject to a KL constraint for the additive reward model, we can either (i) sample $m$ symbols i.i.d. from $P_{\lambda'}^*$ or (ii) sample a block of size $m$ from $P^{*m}_{\lambda}$. As long as $\lambda'=m\lambda$, both solutions render exactly the same reward and KL-divergence relative to $P^{\otimes m}$. However, as we show in the next subsections, this is not the case when we approximate this trade-off curve via BoN or Soft Best-of-$n$ sampling! A significant performance gap exists between applying Best-of-$n$ at a symbol or block level. Next, we analyze why block-level BoN sampling becomes inefficient as the block length increases. 

\noindent \textbf{B. Blockwise Best-of-$n$ sampling.} Let $\mathbb{E}_P[r(X)] = \mu$ be the average reward under the base distribution. For a sequence of length $m$, we analyze the probability of finding rewards significantly above this mean through blockwise sampling.

\begin{prop} 
    Let $Y^m$ be a sequence produced via best-of-$n$ sampling of $n$ sequences drawn from $P_X^{\otimes{m}}$. Then, for $\alpha>0$,
    \begin{align}
        \Pr\left(r(Y^m) \geq \mu + \sqrt{\frac{\alpha \log n}{2m}} \right) \\
        \leq 1-\left(1- n^{-\alpha}\right)^n \xrightarrow{n \to \infty} 
        \begin{cases}
            1&\mbox{if } 0<\alpha<1,\\
            1-1/e &\mbox{if } \alpha = 1,\\
            0 &\mbox{if } \alpha>1.
        \end{cases}
    \end{align}
\end{prop}
\begin{proof}
    Let $X_{i,j}$ be drawn i.i.d. from $P_X$ for $i\in [m]$ and $j\in [n]$, and $Z_j \triangleq r(X_{1,j},\dots,X_{m,j})$ be the reward of the $j$-th sequence. Then for $t > 0$
    \begin{align*} \Pr\left(\max_{j\in[n]} Z_j \geq t\right) &= 1 - [\Pr(Z_1 < t)]^n\\
    &\leq 1 - \left[1 - \exp\left(-2m(t-\mu)^2\right)\right]^n,
    \end{align*}
where the inequality follows from Hoeffding's inequality. The result follows from choosing $t=\mu +\sqrt{\frac{\alpha \log n}{2m}}$.
\end{proof}

This result reveals a fundamental limitation of BoN: to achieve a gain of $\epsilon$ over the average reward $\mu$ using best-of-n sampling, we need $n = e^{O(m\epsilon^2)}$ samples. Intuitively, this exponential dependence on sequence length arises because long sequences tend to concentrate around their mean reward (by the law of large numbers), requiring an exponentially large number of samples $n$ in a Best-of-$n$ sampling strategy to achieve a meaningful gain in reward.

In contrast, consider the symbolwise BoN sampling strategy in Definition \ref{defn:best-of-n} applied to a sequence of $m$ symbols: we sample each of the $m$ symbols $n$ times, keeping the ones with the highest reward value. For simplicity, assume that $r:\mathcal{X} \to \{0,1\}$. In this case, using the notation in Definition  \ref{defn:best-of-n}, the expected reward of the $i$-th symbol is
\begin{equation}
    \mathbb{E}[r(Y_i)] = \Pr\left(r(Y_i\right) = 1) =  1-(1-\mu)^n\triangleq p_n.
\end{equation}
Thus, using a Chernoff bound \cite[Thm. 4]{chung2006concentration} for $\delta > 0$,
\begin{equation}
    \Pr\left(\frac{1}{m}\sum_{i=1}^m r(Y_i) \leq p_n-\delta/m\right) \leq e^{-\frac{\delta^2}{2}m p_n}.
\end{equation}
Note that $\mu \leq p_n \leq 1$, so to achieve a reward of $p_n - \delta/m$ with probability $\tau$ only requires $m \geq -2\log(1-\tau)/\delta^2\mu$ which is independent of $n$. For any $\delta$, we can pick an $m$ sufficiently large such that we achieve a reward with probability $\tau$ and that $\delta/m$ is arbitrarily close to 0. Then, to achieve an $\epsilon$ gain over the average reward $\mu$ requires $n \geq \log(1-\mu-\epsilon)/\log(1-\mu)$. Intuitively, if at each symbol we can pick the reward-maximizing option out of a set of samples, we will very quickly achieve or beat $\mu$ over $m$ symbols. However, the amount by which we beat $\mu$ depends on the number of samples $n$, and $\mu$. In this simple reward setting, symbolwise sampling decouples $n$ and $m$, providing a dramatic improvement over the exponential dependence of $n$ on $m$ found in blockwise sampling.

\noindent \textbf{C. Blockwise Soft Best-of-$n$.} We now analyze the performance of Soft Best-of-$n$ in the blockwise additive reward model.

\begin{defn}\label{defn:n-tilt-blockwise}
For a fixed integer $k \geq 1$ and $\lambda > 0$, the \emph{blockwise Soft Best-of-$n$ sampling} strategy consists of:
\begin{enumerate}
    \item Draw $X^m_1, \ldots, X^m_n$ i.i.d.\ from $P^{\otimes m}_X$;
    \item Draw $Z$ from $\{1, \ldots, n\}$ with distribution
   \begin{equation}
        \Pr(Z=i) = \frac{e^{r(X^m_i)/\lambda}}{\sum_{j=1}^n e^{r(X^m_j)/\lambda}};
    \end{equation}
    \item Return $Y = X_Z$.
\end{enumerate}
\end{defn}
\begin{restatable}{thm}{ThmKLAdditive}\label{thm:KL-additive-reward}
For $\lambda > 0$ and integers $n,m \geq 1$, let $P_X^{\otimes m}$ be an i.i.d. source over $\mathcal{X}^m$ and $r: \mathcal{X} \to [0,1]$ be a reward function. Then, for the blockwise Soft Best-of-$n$ sampling strategy:
\begin{equation*}
\KL(P^{*m}_{\lambda}\|P^{m}_{n,\lambda}) \leq \log\left(1 + \frac{1}{n}\left(\frac{\mathbb{E}_{P_X^{\otimes m}}[e^{2r_\lambda(X)/m}]}{\mathbb{E}_{P_X^{\otimes m}}[e^{r_\lambda(X)/m}]^2} - 1\right)^m\right)
\end{equation*}
Moreover, if $0 \leq r(x) \leq 1$ for all $x \in \mathcal{X}$, then $
\KL(P^*_\lambda \| P_{n,\lambda}) \leq \frac{1}{n}\sinh\left(\frac{1}{m\lambda}\right)^{2m}$.
\end{restatable}

\begin{proof}[Proof sketch of Theorem \ref{thm:KL-additive-reward}]
The proof follow the same steps of the proof of Theorem \ref{thm:KL-bound} but adapted for sequences of length $m$ drawn i.i.d. from $P_X^{\otimes m}$.
\end{proof}

\vspace{-5pt}
\begin{restatable}{cor}{CorLambdaAdditive}
\label{cor:lambda_additivereward}
For $\epsilon > 0$, if $0 \leq r(x) \leq 1$ and $\lambda \geq \frac{2}{\log(1 + 4(n\epsilon)^{1/m})}$, then $\KL(P^{*m}_{\lambda}\|P^{m}_{n,\lambda}) \leq \epsilon$.
\end{restatable}

We analyze how many samples are needed to achieve a target KL divergence when using different block lengths for Soft Best-of-$n$. This analysis reveals a fundamental tradeoff between the symbolwise ($m=1$) and blockwise ($m>1$) approaches. Consider the simplified bound in \eqref{eq:sinh_bound}. To achieve $\KL\left(P^*_{\lambda'} \|P^{m}_{n,\lambda} \right)'\leq \epsilon$, we need the previous bound to be at most $\epsilon$. For the symbolwise sampling, we fix a target $\lambda'$ and recall that the relationship between  $\lambda'$ and $\lambda$ for the blockwise sampling must satisfy $\lambda'=\lambda m$ to ensure equivalent optimal distributions $P^*_{\lambda'} = P^{*m}_\lambda$. From Corollary \ref{cor:lambda_additivereward}, if \begin{equation}
    \label{eq:lambda-match}
    \lambda' = \lambda m \geq \frac{2}{\log(1 + 4(n\epsilon)^{1/m})} \iff n \geq \frac{1}{\epsilon}\left(\frac{e^{\frac{2}{\lambda'}} - 1}{4}\right)^m,
\end{equation}
then $\KL\left(P^*_{\lambda'} \|P^{m}_{n,\lambda} \right)'\leq \epsilon.$ 
Consequently, $n=O\left( \frac{1}{\epsilon}e^{cm/\lambda'}\right),$ and this condition cannot be improved in general due to our lower bound. This shows that to achieve the same ``operating point'' of KL-reward tradeoff (determined by $\lambda'$), tilting a block of size $m$ requires exponentially more samples $n$ to approach the optimal Soft Best-of-$n$ distribution $P^*_{\lambda'}$.

This reveals a fundamental tension in sampling-based alignment strategies. While symbolwise sampling requires exponentially fewer samples to achieve a given point on the reward-KL curve, it necessitates evaluating the reward model at every symbol generation step. In practical applications with expensive reward models, this per-symbol evaluation can become prohibitively costly. On the other hand, blockwise sampling requires exponentially more samples to approach the optimal KL-reward distribution. 

\section{Future Work}
Future directions for this work include the implementation of Soft Best-of-$n$ sampling in real-world generative models, with analysis of resource consumption trade-offs for $n,m$ and $\lambda$. Our sample complexity analysis motivates the development of further controlled decoding strategies (e.g.,\cite{mudgalcontrolled}) that balance competing factors through hybrid approaches using dynamic block lengths or adaptive sampling to reduce computational overhead.  
Finally, this framework can be expanded to more sophisticated and realistic reward models, such as those exhibiting temporal dependencies or causal structures.


\bibliographystyle{IEEEtran}
\bibliography{bib}

\begin{thebibliography}{10}
\providecommand{\url}[1]{#1}
\csname url@samestyle\endcsname
\providecommand{\newblock}{\relax}
\providecommand{\bibinfo}[2]{#2}
\providecommand{\BIBentrySTDinterwordspacing}{\spaceskip=0pt\relax}
\providecommand{\BIBentryALTinterwordstretchfactor}{4}
\providecommand{\BIBentryALTinterwordspacing}{\spaceskip=\fontdimen2\font plus
\BIBentryALTinterwordstretchfactor\fontdimen3\font minus \fontdimen4\font\relax}
\providecommand{\BIBforeignlanguage}[2]{{%
\expandafter\ifx\csname l@#1\endcsname\relax
\typeout{** WARNING: IEEEtran.bst: No hyphenation pattern has been}%
\typeout{** loaded for the language `#1'. Using the pattern for}%
\typeout{** the default language instead.}%
\else
\language=\csname l@#1\endcsname
\fi
#2}}
\providecommand{\BIBdecl}{\relax}
\BIBdecl

\bibitem{hadfield2016cooperative}
D.~Hadfield-Menell, S.~J. Russell, P.~Abbeel, and A.~Dragan, ``Cooperative inverse reinforcement learning,'' \emph{Advances in neural information processing systems}, vol.~29, 2016.

\bibitem{leike2018scalable}
J.~Leike, D.~Krueger, T.~Everitt, M.~Martic, V.~Maini, and S.~Legg, ``Scalable agent alignment via reward modeling: a research direction,'' \emph{arXiv preprint arXiv:1811.07871}, 2018.

\bibitem{bender2021dangers}
E.~M. Bender, T.~Gebru, A.~McMillan-Major, and S.~Shmitchell, ``On the dangers of stochastic parrots: Can language models be too big?'' in \emph{Proceedings of the 2021 ACM conference on fairness, accountability, and transparency}, 2021, pp. 610--623.

\bibitem{bommasani2021opportunities}
R.~Bommasani, D.~A. Hudson, E.~Adeli, R.~Altman, S.~Arora, S.~von Arx, M.~S. Bernstein, J.~Bohg, A.~Bosselut, E.~Brunskill \emph{et~al.}, ``On the opportunities and risks of foundation models,'' \emph{arXiv preprint arXiv:2108.07258}, 2021.

\bibitem{csiszar2003information}
I.~Csisz{\'a}r and F.~Matus, ``Information projections revisited,'' \emph{IEEE Transactions on Information Theory}, vol.~49, no.~6, pp. 1474--1490, 2003.

\bibitem{csiszar1975divergence}
I.~Csisz{\'a}r, ``I-divergence geometry of probability distributions and minimization problems,'' \emph{The annals of probability}, pp. 146--158, 1975.

\bibitem{ouyang2022training}
L.~Ouyang, J.~Wu, X.~Jiang, D.~Almeida, C.~Wainwright, P.~Mishkin, C.~Zhang, S.~Agarwal, K.~Slama, A.~Ray \emph{et~al.}, ``Training language models to follow instructions with human feedback,'' \emph{Advances in neural information processing systems}, vol.~35, pp. 27\,730--27\,744, 2022.

\bibitem{chentsov1968nonsymmetrical}
N.~N. Chentsov, ``Nonsymmetrical distance between probability distributions, entropy and the theorem of pythagoras,'' \emph{Mathematical notes of the Academy of Sciences of the USSR}, vol.~4, no.~3, pp. 686--691, 1968.

\bibitem{csiszar1995generalized}
I.~Csisz{\'a}r, ``Generalized projections for non-negative functions,'' in \emph{Proceedings of 1995 IEEE International Symposium on Information Theory}.\hskip 1em plus 0.5em minus 0.4em\relax IEEE, 1995, p.~6.

\bibitem{alghamdi2020model}
W.~Alghamdi, S.~Asoodeh, H.~Wang, F.~P. Calmon, D.~Wei, and K.~N. Ramamurthy, ``Model projection: Theory and applications to fair machine learning,'' in \emph{2020 IEEE International Symposium on Information Theory (ISIT)}.\hskip 1em plus 0.5em minus 0.4em\relax IEEE, 2020, pp. 2711--2716.

\bibitem{nakayama1992efficient}
M.~K. Nakayama, ``Efficient methods for generating some exponentially tilted random variates,'' in \emph{Proceedings of the 24th conference on Winter simulation}, 1992, pp. 603--608.

\bibitem{hofert2011sampling}
M.~Hofert, ``Sampling exponentially tilted stable distributions,'' \emph{ACM Transactions on Modeling and Computer Simulation (TOMACS)}, vol.~22, no.~1, pp. 1--11, 2011.

\bibitem{fuh2024efficient}
C.-D. Fuh and C.-J. Wang, ``Efficient exponential tilting with applications,'' \emph{Statistics and Computing}, vol.~34, no.~2, p.~65, 2024.

\bibitem{imbens1998information}
G.~W. Imbens, R.~H. Spady, and P.~Johnson, ``Information theoretic approaches to inference in moment condition models,'' \emph{Econometrica}, vol.~66, no.~2, pp. 333--357, 1998.

\bibitem{kitamura1997information}
Y.~Kitamura and M.~Stutzer, ``An information-theoretic alternative to generalized method of moments estimation,'' \emph{Econometrica}, vol.~65, no.~4, pp. 861--874, 1997.

\bibitem{sadowsky1990large}
J.~S. Sadowsky and J.~A. Bucklew, ``On large deviations theory and asymptotically efficient monte carlo estimation,'' \emph{IEEE transactions on Information Theory}, vol.~36, no.~3, pp. 579--588, 1990.

\bibitem{siegmund1976importance}
\BIBentryALTinterwordspacing
D.~Siegmund, ``Importance sampling in the monte carlo study of sequential tests,'' \emph{The Annals of Statistics}, vol.~4, no.~4, pp. 673--684, 1976. [Online]. Available: \url{https://projecteuclid.org/euclid.aos/1176343541}
\BIBentrySTDinterwordspacing

\bibitem{asmussen2007stochastic}
S.~Asmussen and P.~W. Glynn, \emph{Stochastic simulation: algorithms and analysis}.\hskip 1em plus 0.5em minus 0.4em\relax Springer, 2007, vol.~57.

\bibitem{christiano2017deep}
P.~F. Christiano, J.~Leike, T.~Brown, M.~Martic, S.~Legg, and D.~Amodei, ``Deep reinforcement learning from human preferences,'' \emph{Advances in neural information processing systems}, vol.~30, 2017.

\bibitem{stiennon2020learning}
N.~Stiennon, L.~Ouyang, J.~Wu, D.~Ziegler, R.~Lowe, C.~Voss, A.~Radford, D.~Amodei, and P.~F. Christiano, ``Learning to summarize with human feedback,'' \emph{Advances in Neural Information Processing Systems}, vol.~33, pp. 3008--3021, 2020.

\bibitem{rafailov2024direct}
R.~Rafailov, A.~Sharma, E.~Mitchell, C.~D. Manning, S.~Ermon, and C.~Finn, ``Direct preference optimization: Your language model is secretly a reward model,'' \emph{Advances in Neural Information Processing Systems}, vol.~36, 2024.

\bibitem{gao2023scaling}
L.~Gao, J.~Schulman, and J.~Hilton, ``Scaling laws for reward model overoptimization,'' in \emph{International Conference on Machine Learning}.\hskip 1em plus 0.5em minus 0.4em\relax PMLR, 2023, pp. 10\,835--10\,866.

\bibitem{mudgalcontrolled}
S.~Mudgal, J.~Lee, H.~Ganapathy, Y.~Li, T.~Wang, Y.~Huang, Z.~Chen, H.-T. Cheng, M.~Collins, T.~Strohman \emph{et~al.}, ``Controlled decoding from language models,'' in \emph{Forty-first International Conference on Machine Learning}, 2024.

\bibitem{yang2024asymptotics}
J.~Q. Yang, S.~Salamatian, Z.~Sun, A.~T. Suresh, and A.~Beirami, ``Asymptotics of language model alignment,'' \emph{arXiv preprint arXiv:2404.01730}, 2024.

\bibitem{mroueh2024information}
Y.~Mroueh, ``Information theoretic guarantees for policy alignment in large language models,'' \emph{arXiv preprint arXiv:2406.05883}, 2024.

\bibitem{schulman2017proximal}
J.~Schulman, F.~Wolski, P.~Dhariwal, A.~Radford, and O.~Klimov, ``Proximal policy optimization algorithms,'' \emph{arXiv preprint arXiv:1707.06347}, 2017.

\bibitem{gui2024bonbon}
L.~Gui, C.~Garbacea, and V.~Veitch, ``{BoNBoN Alignment for Large Language Models and the Sweetness of Best-of-n Sampling},'' in \emph{The Thirty-eighth Annual Conference on Neural Information Processing Systems}, 2024.

\bibitem{jinnai2024regularized}
Y.~Jinnai, T.~Morimura, K.~Ariu, and K.~Abe, ``Regularized best-of-n sampling to mitigate reward hacking for language model alignment,'' in \emph{ICML 2024 Workshop on Models of Human Feedback for AI Alignment}, 2024.

\bibitem{ichihara2025evaluation}
Y.~Ichihara, Y.~Jinnai, T.~Morimura, K.~Abe, K.~Ariu, M.~Sakamoto, and E.~Uchibe, ``Evaluation of best-of-n sampling strategies for language model alignment,'' \emph{Transactions on Machine Learning Research}, 2025.

\bibitem{amini2024variational}
A.~Amini, T.~Vieira, and R.~Cotterell, ``Variational best-of-n alignment,'' \emph{arXiv preprint arXiv:2407.06057}, 2024.

\bibitem{qiu2024treebon}
J.~Qiu, Y.~Lu, Y.~Zeng, J.~Guo, J.~Geng, H.~Wang, K.~Huang, Y.~Wu, and M.~Wang, ``Treebon: Enhancing inference-time alignment with speculative tree-search and best-of-n sampling,'' \emph{arXiv preprint arXiv:2410.16033}, 2024.

\bibitem{zhang2024accelerating}
R.~Zhang, M.~Haider, M.~Yin, J.~Qiu, M.~Wang, P.~Bartlett, and A.~Zanette, ``Accelerating best-of-n via speculative rejection,'' in \emph{2nd Workshop on Advancing Neural Network Training: Computational Efficiency, Scalability, and Resource Optimization (WANT@ ICML 2024)}, 2024.

\bibitem{nakano2021webgpt}
R.~Nakano, J.~Hilton, S.~Balaji, J.~Wu, L.~Ouyang, C.~Kim, C.~Hesse, S.~Jain, V.~Kosaraju, W.~Saunders \emph{et~al.}, ``Webgpt: Browser-assisted question-answering with human feedback,'' \emph{arXiv preprint arXiv:2112.09332}, 2021.

\bibitem{hilton2022measuringgoodhart}
\BIBentryALTinterwordspacing
J.~Hilton, P.~Clark \emph{et~al.}, ``Measuring goodhart’s law: Towards an evaluation framework for open-ended generative models,'' OpenAI Blog, 2022, accessed: 2025-01-30. [Online]. Available: \url{https://openai.com/index/measuring-goodharts-law}
\BIBentrySTDinterwordspacing

\bibitem{beirami2024theoretical}
A.~Beirami, A.~Agarwal, J.~Berant, A.~D'Amour, J.~Eisenstein, C.~Nagpal, and A.~T. Suresh, ``Theoretical guarantees on the best-of-n alignment policy,'' \emph{arXiv preprint arXiv:2401.01879}, 2024.

\bibitem{bhatia2000better}
R.~Bhatia and C.~Davis, ``A better bound on the variance,'' \emph{The american mathematical monthly}, vol. 107, no.~4, pp. 353--357, 2000.

\bibitem{berry1941accuracy}
A.~C. Berry, ``The accuracy of the gaussian approximation to the sum of independent variates,'' \emph{Transactions of the american mathematical society}, vol.~49, no.~1, pp. 122--136, 1941.

\bibitem{esseen1942liapounoff}
C.~Esseen, \emph{On the Liapounoff Limit of Error in the Theory of Probability}.\hskip 1em plus 0.5em minus 0.4em\relax Almqvist \& Wiksell, 1942.

\bibitem{cover1999elements}
T.~M. Cover, \emph{Elements of information theory}.\hskip 1em plus 0.5em minus 0.4em\relax John Wiley \& Sons, 1999.

\bibitem{chung2006concentration}
F.~Chung and L.~Lu, ``Concentration inequalities and martingale inequalities: a survey,'' \emph{Internet mathematics}, vol.~3, no.~1, pp. 79--127, 2006.

\bibitem{touvron2023llama}
H.~Touvron, L.~Martin, K.~Stone, P.~Albert, A.~Almahairi, Y.~Babaei, N.~Bashlykov, S.~Batra, P.~Bhargava, S.~Bhosale \emph{et~al.}, ``Llama 2: Open foundation and fine-tuned chat models,'' \emph{arXiv preprint arXiv:2307.09288}, 2023.

\bibitem{sessa2024bond}
P.~G. Sessa, R.~Dadashi, L.~Hussenot, J.~Ferret, N.~Vieillard, A.~Ram{\'e}, B.~Shariari, S.~Perrin, A.~Friesen, G.~Cideron \emph{et~al.}, ``Bond: Aligning llms with best-of-n distillation,'' \emph{arXiv preprint arXiv:2407.14622}, 2024.

\bibitem{wei2022chain}
J.~Wei, X.~Wang, D.~Schuurmans, M.~Bosma, F.~Xia, E.~Chi, Q.~V. Le, D.~Zhou \emph{et~al.}, ``Chain-of-thought prompting elicits reasoning in large language models,'' \emph{Advances in neural information processing systems}, vol.~35, pp. 24\,824--24\,837, 2022.

\bibitem{lightman2023let}
H.~Lightman, V.~Kosaraju, Y.~Burda, H.~Edwards, B.~Baker, T.~Lee, J.~Leike, J.~Schulman, I.~Sutskever, and K.~Cobbe, ``Let's verify step by step,'' in \emph{The Twelfth International Conference on Learning Representations}, 2023.

\bibitem{uesato2022solving}
J.~Uesato, N.~Kushman, R.~Kumar, F.~Song, N.~Siegel, L.~Wang, A.~Creswell, G.~Irving, and I.~Higgins, ``Solving math word problems with process-and outcome-based feedback,'' \emph{arXiv preprint arXiv:2211.14275}, 2022.

\bibitem{puri2025probabilistic}
I.~Puri, S.~Sudalairaj, G.~Xu, K.~Xu, and A.~Srivastava, ``A probabilistic inference approach to inference-time scaling of llms using particle-based monte carlo methods,'' \emph{arXiv preprint arXiv:2502.01618}, 2025.

\end{thebibliography}

\onecolumn
\appendix

\paragraph{Roadmap.}
This appendix provides detailed proofs and technical analysis to support the main results of the paper. We begin in Section~\ref{sec:appendix_related} by expanding on the related literature and positioning our work within the broader context. In Section~\ref{sec:proof_upperbound}, we prove Theorem~\ref{thm:KL-bound}, establishing an $O(1/n)$ upper bound on the KL divergence for Soft Best-of-$n$ sampling. Next, Section~\ref{sec:lambda_lower_bound} develops a lower bound on the convergence rate with respect to the tilting parameter $\lambda$. Section~\ref{sec:sequence_upperbound} then demonstrates the sharpness of Theorem~\ref{thm:KL-bound} by constructing a binary example that attains the $O(1/n)$ rate. 

In Section~\ref{sec:relative_lower_bound}, we derive a relative reward bound, showing that the expected reward of Soft Best-of-$n$ sampling converges to that of the optimal tilted distribution at the same $O(1/n)$ rate. Finally, Section~\ref{sec:appendix_blockwise_sampling} focuses on blockwise sampling, providing key lemmas and analysis for sequences of tokens rather than individual tokens, and Sections~\ref{sec:proof_upperbound_blockwise} and \ref{sec:corollary_lambda_blockwise} extends these results by deriving specific corollaries for the blockwise case.

\subsection{Additional related work.} \label{sec:appendix_related}

\textbf{Alignment.} LLM alignment is generally done via Reinforcement Learning from Human Feedback (RLHF) \cite{christiano2017deep}, which uses reinforcement learning to steer an LLM's sampling distribution to maximize a reward signal, oftentimes with a KL penalty to enforce similarity to the pre-trained LLM's behavior. This can be done with a separate reward model such as in Proximal Policy Optimization \cite{schulman2017proximal} or directly via Direct Policy Optimization \cite{rafailov2024direct}. 

\textbf{Best-of-$n$.} Alternatively, Best-of-$n$ sampling is a simple inference-time approach for alignment \cite{stiennon2020learning, nakano2021webgpt, hilton2022measuringgoodhart} and has even been used in the development of industry models such as Llama 2 \cite{touvron2023llama}. There have been various methodological improvements on BoN sampling such as variational BoN, \cite{amini2024variational}, TreeBoN \cite{qiu2024treebon}, and Speculative BoN \cite{zhang2024accelerating} which all try to reduce the computational cost of sampling $n$ sequences. Other works such as \cite{gui2024bonbon} and \cite{sessa2024bond} propose training techniques to fine-tune a language model to sample from the BoN distribution while preserving model performance. 

\textbf{Prior Results.} Early works proposed an analytical formula for the KL divergence between BoN sampling and the reference model \cite{stiennon2020learning, hilton2022measuringgoodhart}, but recent work demonstrated that this formula is in fact an upper bound and overly-conservative characterization of KL \cite{beirami2024theoretical, mroueh2024information}, indicating that BoN sampling might perform better than expected in terms of KL-reward tradeoffs. Furthermore, \cite{yang2024asymptotics} demonstrated the asymptotic equivalence between BoN and KL-constrained reinforcement learning under certain conditions, providing theoretical justification for BoN's strong empirical performance. Empirically, \cite{gao2023scaling} and \cite{hilton2022measuringgoodhart} found that the KL-reward tradeoff for both RL and BoN methods scaled as $\sqrt{D_{\mathsf{KL}}}$ between the optimized and base distribution, and \cite{mroueh2024information} demonstrated that this relationship is an information-theoretic upper bound. 

\textbf{Process Reward Modeling in Reasoning.} There has been a recent surge of interest in properly developing models with reasoning capabilities by encouraging long generations where models can work through complex problems step-by-step \cite{wei2022chain}. To encourage proper reasoning, process reward models (PRMs) have emerged as a way to provide reward at each reasoning step of a model, and have been shown to improve reasoning capabilities over rewarding outcomes only \cite{lightman2023let, uesato2022solving}. PRMs benefit from exhibiting similar properties to the additive reward model mentioned above; intuitively, by rewarding on a finer-grained scale (during each reasoning step instead of rewarding the final output), we can sample ``symbolwise'' (where a symbol is a reasoning step) to more efficiently achieve a gain in reward. The work most similar to ours by Puri et al. \cite{puri2025probabilistic} uses a sampling approach similar to Soft Best-of-$n$ for PRMs. However, where our method would sample \textit{once} from the soft distribution over outputs for each generation, and then generate $n$ times from that chosen sample for the next Soft Best-of-$n$ round, \cite{puri2025probabilistic} instead maintains $n$ particles throughout generation and at each step samples with replacement $n$ times from the soft distribution to continue generation from these $n$ particles.

\subsection{Proof of the upper bound}\label{sec:proof_upperbound}

The proof of Theorem \ref{thm:KL-bound}, restated here for the sake of completeness, relies on Lemma \ref{lem:p-char} that provides a closed-form expression for $P_{n,\lambda}$. For simplicity, we denote $r_\lambda(x)\triangleq r(x)/\lambda.$ 

\ThmKLUpper*

We start by proving this lemma.

\LemSoftBon*

\begin{proof}
Following Definition 2, we can express $P_{n,\lambda}(y)$ by considering all possible ways that $y$ could be selected as the output. This occurs when some $X_i = y$ and that index $i$ is selected by $Z$:

\begin{align*}
P_{n,\lambda}(y) &= \sum_{i=1}^n \Pr(Y = y, Z = i) \\
&= \sum_{i=1}^n \Pr(X_i = y, Z = i) \\
&= \sum_{i=1}^n \Pr(X_i = y)\Pr(Z = i|X_i = y) \\
&= n\Pr(X_1 = y)\Pr(Z = 1|X_1 = y) \quad \text{(by symmetry)}
\end{align*}

Now, to compute $\Pr(Z = 1|X_1 = y)$, we need to consider all possible values that the remaining samples $X_2, \ldots, X_n$ could take. By the law of total probability:

\begin{align*}
&\Pr(Z = 1|X_1 = y) \\
&= \sum_{x_2 \in \mathcal{X}} \cdots \sum_{x_n \in \mathcal{X}} \Pr(Z = 1|X_1 = y, X_2 = x_2, \ldots, X_n = x_n)\\ &\times \Pr(X_2 = x_2, \ldots, X_n = x_n) \\
&= \sum_{x_2 \in \mathcal{X}} \cdots \sum_{x_n \in \mathcal{X}} \frac{e^{r(y)/\lambda}}{e^{r(y)/\lambda} + \sum_{i=2}^n e^{r(x_i)/\lambda}} \times \prod_{i=2}^n P(x_i) \\
&= \mathbb{E}\left[\frac{e^{r(y)/\lambda}}{e^{r(y)/\lambda} + \sum_{i=2}^n e^{r(X_i)/\lambda}}\right]
\end{align*}

where the last line follows from the fact that $X_2, \ldots, X_n$ are i.i.d.\ samples from $P$. Therefore:
\begin{align*}
P_{n,\lambda}(y) &= nP(y)\mathbb{E}\left[\frac{e^{r(y)/\lambda}}{e^{r(y)/\lambda} + \sum_{i=2}^n e^{r(X_i)/\lambda}}\right] \\
&= P(y)e^{r(y)/\lambda} \times \mathbb{E}\left[\frac{1}{\frac{1}{n}(e^{r(y)/\lambda} + \sum_{i=1}^{n-1} e^{r(X_i)/\lambda})}\right],
\end{align*}
expectation is with respect to the joint distribution of $(X_2, \dots, X_n)$. The inequality follows from Jensen's inequality $
\mathbb{E}[1/X] \geq 1/\mathbb{E}[X]$.
\end{proof}

We can now finally prove Theorem \ref{thm:KL-bound}.

\ThmKLUpper*

\begin{proof}[Proof of Theorem \ref{thm:KL-bound}:]
    For $P_\lambda^*$ given in \eqref{eq:Pstar-def} and  $ P_{n,\lambda}$ given in \eqref{eq:prob-exact}:
    \begin{align}
         \KL\left(P^*_\lambda \| P_{n,\lambda}\right) = \sum_x P^*_\lambda(x) \log \left(\frac{ P^*_\lambda(x) }{ P_{n,\lambda}(x)} \right) \\
         \leq \sum_x \frac{P(x)e^{r_\lambda(x)}}{\mathbb{E}_P[e^{r_\lambda(X)}]} \log\left( \frac{1/\mathbb{E}_P[e^{r_\lambda(X)}]}{1/\left(\frac{1}{n}e^{r_\lambda(x)}+\frac{n-1}{n}\mathbb{E}_P\left[e^{r_\lambda(X)} \right]\right)} \right)
    \end{align}
    where the inequality follows from plugging in definitions for $P_\lambda^*$ and $ P_{n,\lambda}$  and applying \eqref{eq:prob-lb}. Rearranging the terms, we have:
    \begin{align*}
        \KL\left(P^*_\lambda \| P_{n,\lambda}\right) &\leq  \sum_x  \frac{P(x)e^{r_\lambda(x)}}{\mathbb{E}_P[e^{r_\lambda(X)}]} \log\left( \frac{1}{n} \frac{e^{r_\lambda(x)}}{\mathbb{E}_P[e^{r_\lambda(X)}]}+ \frac{n-1}{n}\right) \\
        &\leq  \log \left[ \frac{1}{n}\left(  \sum_x  \frac{P(x)e^{2r_\lambda(x)}}{\mathbb{E}_P[e^{r_\lambda(X)}]^2}\right) +  \frac{n-1}{n}  \right] \tag*{\text{(Jensen's inequality)}} \\
        & =  \log \left[ \frac{1}{n}\left(    \frac{\mathbb{E}_P[e^{2r_\lambda(X)}]}{\mathbb{E}_P[e^{r_\lambda(X)}]^2}-1\right) +  1  \right]. 
    \end{align*}
    Equation \eqref{eq:KL-bound} follows by plugging in the definition of $ \mathsf{CV}\left(e^{r_\lambda(X)}\right)$. Since $\log (1+x)\leq x$ (assuming base $e$), then $\KL\left(P^*_\lambda \| P_{n,\lambda}\right)= O(1/n).$ 

    Now assume that the reward function $r(x)$ is bounded by 1, and denote $$\sigma^2\triangleq \mathbb{E}_P[e^{2r_\lambda(X)}]-\mathbb{E}_P[e^{r_\lambda(X)}]^2,~\mu = \mathbb{E}_P[e^{r_\lambda(X)}],~  c\triangleq e^{1/\lambda}. $$ From the Bhatia-Davis inequality \cite{bhatia2000better}, $\sigma^2 \leq \left(c-\mu\right)\left(\mu-1 \right)$. Since $\mu\in [1,c]$, a simple maximization yields
    \begin{align*}
       \mathsf{CV}\left(e^{r_\lambda(X)}\right)^2\leq  \max_{\mu \in [1,c]} \frac{\left(c-\mu\right)\left(\mu-1 \right)}{\mu^2} = \frac{(c-1)^2}{4c}.
    \end{align*}
    Consequently,
    \begin{align*}
        \KL\left(P^*_\lambda \| P_{n,\lambda}\right) &\leq \log\left(1+ \frac{1}{n}\frac{(e^{1/\lambda}-1)^2}{4e^{1/\lambda}} \right) \\
        &\leq \frac{1}{n}\frac{(e^{1/\lambda}-1)^2}{4e^{1/\lambda}} \\
        &= \frac{1}{n}\sinh \left(\frac{1}{2\lambda}\right)^2.
    \end{align*}
\end{proof}

\subsection{Proof of Corollary \ref{cor:lambda_lower_bound}}\label{sec:lambda_lower_bound}
We begin by restating Corollary \ref{cor:lambda_lower_bound}.
\CorLambda*

\begin{proof}

From Theorem \ref{thm:KL-bound}, we know that when $0 \leq r(x) \leq 1$, we have $\KL(P^*_\lambda\|P_{n,\lambda}) \leq \frac{1}{n}\sinh^2(\frac{1}{2\lambda})$. For the corollary's claim to hold, we need this upper bound to be at most $\epsilon$. Therefore, we start by writing:

\begin{equation}
\frac{1}{n}\sinh^2(\frac{1}{2\lambda}) \leq \epsilon
\end{equation}

Multiplying both sides by $n$, we obtain $\sinh^2(\frac{1}{2\lambda}) \leq n\epsilon$. Taking the square root of both sides and noting that since $\lambda > 0$, we have $\frac{1}{2\lambda} > 0$, we can write $\sinh(\frac{1}{2\lambda}) \leq \sqrt{n\epsilon}$. Then, applying the inverse hyperbolic sine ($\arcsinh$) to both sides yields:

\[\frac{1}{2\lambda} \leq \arcsinh(\sqrt{n\epsilon})\]

Now, we can use the fact that for $x > 0$, $\arcsinh(x) = \ln(x + \sqrt{x^2 + 1})$. Therefore:

\[\lambda \geq \frac{1}{2\ln(\sqrt{n\epsilon} + \sqrt{n\epsilon + 1})}\]

Note that for $x > 0$, $x + \sqrt{x^2 + 1} \geq 2x$ since
\begin{align*}
(x + \sqrt{x^2 + 1})^2 &= x^2 + 2x\sqrt{x^2 + 1} + (x^2 + 1) \\
&= 2x^2 + 2x\sqrt{x^2 + 1} + 1 \\
&\geq 2x^2 + 2x^2 + 1 \quad \text{(since $\sqrt{x^2 + 1} > x$ for $x > 0$)} \\
&> 4x^2
\end{align*}

Thus, $\ln(x + \sqrt{x^2 + 1}) \geq \ln(2x)$, which implies:

\[\frac{1}{2\ln(\sqrt{n\epsilon} + \sqrt{n\epsilon + 1})} \leq \frac{1}{\ln(1 + 4n\epsilon)}\]

Thus, $\lambda \geq \frac{1}{\log(1 + 4n\epsilon)}$ is a sufficient condition to ensure $D_{KL}(P^*_\lambda\|P_{n,\lambda}) \leq \epsilon$.
\end{proof}

\subsection{A sequence that attains the upper bound}\label{sec:sequence_upperbound}

In Theorem \ref{thm:KL-bound}, we established that $\KL\left(P_\lambda^* \| P_{n,\lambda}\right) \leq \log\left(1 + \frac{1}{n}\mathrm{CV}\left(e^{r_\lambda(X)}\right)^2\right)$, i.e., $\KL\left(P_\lambda^* \| P_{n,\lambda}\right) = O(1/n)$. We will now establish that this upper bound is, indeed, sharp. Let us construct a simple binary example. Let $\mathcal{X} = \{0,1\}$ with $P(1) = p$ and $P(0) = 1-p$ where $p \in (0,1)$. Define the reward function:

\begin{equation*}
r(x) = \begin{cases}
1 & \text{if } x = 1 \\
0 & \text{if } x = 0
\end{cases}
\end{equation*}

Then:
\begin{align*}
E_P[e^{r_\lambda(X)}] &= pe^{1/\lambda} + (1-p) \\
E_P[e^{2r_\lambda(X)}] &= pe^{2/\lambda} + (1-p)
\end{align*}

Therefore:
\begin{align*}
\text{Var}(e^{r_\lambda(X)}) &= E_P[e^{2r_\lambda(X)}] - (E_P[e^{r_\lambda(X)}])^2 \\
&= pe^{2/\lambda} + (1-p) - (pe^{1/\lambda} + (1-p))^2 \\
&= pe^{2/\lambda} + (1-p) - (p^2e^{2/\lambda} + 2p(1-p)e^{1/\lambda} + (1-p)^2) \\
&= p(1-p)e^{2/\lambda} - 2p(1-p)e^{1/\lambda} + (1-p)(p) \\
&= p(1-p)(e^{1/\lambda} - 1)^2
\end{align*}

Which, in turn, implies
\begin{equation*}
CV(e^{r_\lambda(X)})^2 = \frac{\text{Var}(e^{r_\lambda(X)})}{(E_P[e^{r_\lambda(X)}])^2} = \frac{p(1-p)(e^{1/\lambda} - 1)^2}{(pe^{1/\lambda} + (1-p))^2}
\end{equation*}

From Lemma \ref{lem:p-char}, for $x = 1$:
\begin{align*}
P_{n,\lambda}(1) &= pe^{1/\lambda} E\left[\frac{1}{\frac{1}{n}(e^{1/\lambda} + S_{n-1})}\right] \\
&= npe^{1/\lambda} E\left[\frac{1}{e^{1/\lambda} + S_{n-1}}\right]
\end{align*}

Now, we can write the Taylor expansion for $E\left[\frac{1}{e^{1/\lambda} + S_{n-1}}\right]$. Consider the function $g(x) = \frac{1}{e^{1/\lambda} + x}$ and let $S_{n-1} = \sum_{i=1}^{n-1} e^{r_\lambda(X_i)}$. Note that $g(x)$ is analytic and uniformly bounded by $1/e^{1/\lambda}$ for all $x \geq 0$. The second-order Taylor expansion of $g(x)$ around $\mu = E[S_{n-1}]$ is:
\begin{equation*}
g(x) = g(\mu) + g'(\mu)(x-\mu) + \frac{g''(\mu)}{2}(x-\mu)^2 + R_3(x)
\end{equation*}
The law of large numbers guarantees that for any $\varepsilon > 0$,
\begin{equation}
\Pr\bigl( |S_{n-1} - \mu| \ge \varepsilon \bigr) \to 0 \quad \text{as } n \to \infty.
\end{equation}
Define the high-probability event $
\mathcal{A} = \{\, |S_{n-1} - \mu| \le \varepsilon \,\}.$ On $\mathcal{A}$, $S_{n-1}$ lies within the interval $[\mu - \varepsilon, \mu + \varepsilon]$, a region on which the derivatives of $g$ are well-behaved and the Taylor expansion provides a uniformly accurate approximation (i.e., there exists a constant $C>0$ such that $|R_3(S_{n-1})| \le C |S_{n-1} - \mu|^3)$. On the complementary event $\mathcal{A}^c$, the Taylor expansion might not be as accurate; however, the probability of $\mathcal{A}^c$ decays exponentially in $n$. Therefore, even if the error in the Taylor expansion is large on $\mathcal{A}^c$, its overall contribution to expectations is negligible. Computing the derivatives of $g$:
\begin{align*}
g'(x) &= -\frac{1}{(e^{1/\lambda} + x)^2} \\
g''(x) &= \frac{2}{(e^{1/\lambda} + x)^3} \\
g'''(x) &= -\frac{6}{(e^{1/\lambda} + x)^4}
\end{align*}

Therefore:
\begin{align*}
g(\mu) &= \frac{1}{e^{1/\lambda} + (n-1)(pe^{1/\lambda} + (1-p))} \\
g'(\mu) &= -\frac{1}{(e^{1/\lambda} + (n-1)(pe^{1/\lambda} + (1-p)))^2} \\
g''(\mu) &= \frac{2}{(e^{1/\lambda} + (n-1)(pe^{1/\lambda} + (1-p)))^3}
\end{align*}

Taking expectations:
\begin{align*}
E[g(S_{n-1})] &= g(\mu) + g'(\mu)E[(S_{n-1}-\mu)] \\ & + \frac{g''(\mu)}{2}E[(S_{n-1}-\mu)^2] + E[R_3(S_{n-1})] \\
&= g(\mu) + 0 + \frac{g''(\mu)}{2}\text{Var}(S_{n-1}) + E[R_3(S_{n-1})]
\end{align*}

Substituting the values:
\begin{align*}
E\left[\frac{1}{e^{1/\lambda} + S_{n-1}}\right] &= \frac{1}{e^{1/\lambda} + (n-1)(pe^{1/\lambda} + (1-p))} \\
&+ \frac{(n-1)p(1-p)(e^{1/\lambda} - 1)^2}{(e^{1/\lambda} + (n-1)(pe^{1/\lambda} + (1-p)))^3} \\
&+ E[R_3(S_{n-1})]
\end{align*}

(This last term is definitely small. In particular, smaller than $O(1/n^2)$. We will take care of it later).

A similar calculation for $x = 0$ gives:
\begin{equation*}
P_{n,\lambda}(0) = \frac{n(1-p)}{1 + (n-1)(pe^{1/\lambda} + (1-p))} + O(\frac{1}{n^2})
\end{equation*}

This leads to our four final expressions:

\begin{align*}
P^*_\lambda(1) &= \frac{pe^{1/\lambda}}{pe^{1/\lambda} + (1-p)} \\
P^*_\lambda(0) &= \frac{1-p}{pe^{1/\lambda} + (1-p)} \\
P_{n,\lambda}(1) &= \frac{npe^{1/\lambda}}{e^{1/\lambda} + (n-1)(pe^{1/\lambda} + (1-p))} + O(1/n^2) \\
P_{n,\lambda}(0) &= \frac{n(1-p)}{1 + (n-1)(pe^{1/\lambda} + (1-p))} + O(1/n^2)
\end{align*}

We can now compute the KL divergence

\begin{align*}
&\KL(P^*_\lambda \| P_{n,\lambda}) = P^*_\lambda(1)\log\frac{P^*_\lambda(1)}{P_{n,\lambda}(1)} + P^*_\lambda(0)\log\frac{P^*_\lambda(0)}{P_{n,\lambda}(0)} \\
&= \frac{pe^{1/\lambda}}{pe^{1/\lambda} + (1-p)}\\ & \times \log\left(\frac{pe^{1/\lambda}}{pe^{1/\lambda} + (1-p)} \cdot \frac{e^{1/\lambda} + (n-1)(pe^{1/\lambda} + (1-p))}{npe^{1/\lambda}} + O(1/n^2)\right) \\
&+ \frac{1-p}{pe^{1/\lambda} + (1-p)} \\ & \times \log\left(\frac{1-p}{pe^{1/\lambda} + (1-p)} \cdot \frac{1 + (n-1)(pe^{1/\lambda} + (1-p))}{n(1-p)} + O(1/n^2)\right)
\end{align*}

Let's simplify the terms inside each logarithm. For the first term:
\begin{align*}
&\frac{e^{1/\lambda} + (n-1)(pe^{1/\lambda} + (1-p))}{n(pe^{1/\lambda} + (1-p))} \\
&= \frac{e^{1/\lambda} - (pe^{1/\lambda} + (1-p)) + n(pe^{1/\lambda} + (1-p))}{n(pe^{1/\lambda} + (1-p))} \\
&= 1 + \frac{e^{1/\lambda} - (pe^{1/\lambda} + (1-p))}{n(pe^{1/\lambda} + (1-p))} \\
&= 1 + \frac{(1-p)(e^{1/\lambda} - 1)}{n(pe^{1/\lambda} + (1-p))}
\end{align*}

Similarly for the second term:
\begin{align*}
&\frac{1 + (n-1)(pe^{1/\lambda} + (1-p))}{n(pe^{1/\lambda} + (1-p))} \\
&= 1 + \frac{1 - (pe^{1/\lambda} + (1-p))}{n(pe^{1/\lambda} + (1-p))} \\
&= 1 - \frac{p(e^{1/\lambda} - 1)}{n(pe^{1/\lambda} + (1-p))}
\end{align*}

This leads to 

\begin{align*}
&\KL(P^*_\lambda \| P_{n,\lambda}) = \\
&\frac{pe^{1/\lambda}}{pe^{1/\lambda} + (1-p)}\log\left(1 + \frac{(1-p)(e^{1/\lambda} - 1)}{n(pe^{1/\lambda} + (1-p))} + O(1/n^2)\right) \\
&+ \frac{1-p}{pe^{1/\lambda} + (1-p)}\log\left(1 - \frac{p(e^{1/\lambda} - 1)}{n(pe^{1/\lambda} + (1-p))} + O(1/n^2)\right)
\end{align*}

For the first logarithm, let's denote $a_1 = \frac{(1-p)(e^{1/\lambda} - 1)}{n(pe^{1/\lambda} + (1-p))}$. Then:
\begin{align*}
&\log(1 + a_1 + O(1/n^2)) = (a_1 + O(1/n^2))\\& - \frac{(a_1 + O(1/n^2))^2}{2} + O((a_1 + O(1/n^2))^3) \\
&= a_1 + O(1/n^2) - \frac{a_1^2}{2} - O(a_1/n^2) + O(1/n^3) \\
&= \frac{(1-p)(e^{1/\lambda} - 1)}{n(pe^{1/\lambda} + (1-p))} - \frac{(1-p)^2(e^{1/\lambda} - 1)^2}{2n^2(pe^{1/\lambda} + (1-p))^2} + O(1/n^3)
\end{align*}

Similarly for the second logarithm, let $a_2 = -\frac{p(e^{1/\lambda} - 1)}{n(pe^{1/\lambda} + (1-p))}$:
\begin{align*}
&\log(1 + a_2 + O(1/n^2)) = (a_2 + O(1/n^2))\\& - \frac{(a_2 + O(1/n^2))^2}{2} + O((a_2 + O(1/n^2))^3) \\
&= a_2 + O(1/n^2) - \frac{a_2^2}{2} - O(a_2/n^2) + O(1/n^3) \\
&= -\frac{p(e^{1/\lambda} - 1)}{n(pe^{1/\lambda} + (1-p))} - \frac{p^2(e^{1/\lambda} - 1)^2}{2n^2(pe^{1/\lambda} + (1-p))^2} + O(1/n^3)
\end{align*}

Multiplying by the respective coefficients and combining terms:
\begin{align*}
\KL(P^*_\lambda \| P_{n,\lambda}) &= \frac{pe^{1/\lambda}}{pe^{1/\lambda} + (1-p)} \cdot \frac{(1-p)(e^{1/\lambda} - 1)}{n(pe^{1/\lambda} + (1-p))} \\
&- \frac{pe^{1/\lambda}}{pe^{1/\lambda} + (1-p)} \cdot \frac{(1-p)^2(e^{1/\lambda} - 1)^2}{2n^2(pe^{1/\lambda} + (1-p))^2} \\
&+ \frac{1-p}{pe^{1/\lambda} + (1-p)} \cdot \left(-\frac{p(e^{1/\lambda} - 1)}{n(pe^{1/\lambda} + (1-p))}\right) \\
&- \frac{1-p}{pe^{1/\lambda} + (1-p)} \cdot \frac{p^2(e^{1/\lambda} - 1)^2}{2n^2(pe^{1/\lambda} + (1-p))^2} \\
&+ O(1/n^3)
\end{align*}

The first-order terms (in $1/n$) combine to:
\begin{equation*}
\frac{p(1-p)(e^{1/\lambda} - 1)^2}{n(pe^{1/\lambda} + (1-p))^2}
\end{equation*}

The second-order terms (in $1/n^2$) combine to:
\begin{equation*}
-\frac{p(1-p)(e^{1/\lambda} - 1)^2}{2n^2(pe^{1/\lambda} + (1-p))^2} \cdot \frac{p(1-p)(e^{1/\lambda} - 1)^2}{(pe^{1/\lambda} + (1-p))^2}
\end{equation*}

Therefore:
\begin{equation*}
\KL(P^*_\lambda \| P_{n,\lambda}) = \frac{p(1-p)(e^{1/\lambda} - 1)^2}{n(pe^{1/\lambda} + (1-p))^2} + O(1/n^2)
\end{equation*}

This matches our earlier coefficient of variation calculation:
\begin{equation*}
\frac{1}{n}CV(e^{r_\lambda(X)})^2 = \frac{1}{n}\frac{p(1-p)(e^{1/\lambda} - 1)^2}{(pe^{1/\lambda} + (1-p))^2}
\end{equation*}

\subsection{Proof of the lower bound}\label{sec:proof_lowerbound}

We begin by restating Theorem \ref{thm:KL-lowerbound}.

\ThmKLLower*

Let $$P^*_{\lambda}(x)\triangleq \frac{P(x)e^{r(x)}}{\mathbb{E}_P\left[e^{r(X)} \right]}$$ and $$P_{n,\lambda}(x)\triangleq P(x)e^{r(x)}\times \mathbb{E}_{P^{\otimes (n-1)}}\left[\frac{n}{e^{r(x)}+ \sum_{i=1}^{n-1} e^{r(X_i)}}\right] ,$$
where $r:\mathcal{X}\to [0,a].$ Note that $a$ will be $e^{1/\lambda}$ once we add the dependence on $\lambda$.

Denoting $\mu \triangleq  \mathbb{E}_P\left[e^{r(X)} \right]$, we have
\begin{align*}
    \mathsf{TV}\left(P^*_{\lambda} \|P_{n,\lambda} \right) &= \frac{1}{2}\sum_x \left|P^*_{\lambda}(x) - P_{n,\lambda}(x) \right|\\
    &= \frac{1}{2}\sum_x  \frac{P(x)e^{r(x)}}{\mu} \times  \left|1 -  \mathbb{E}_{P^{\otimes (n-1)}}\left[\frac{1}{\frac{e^{r(x)}}{n \mu}+ \left(\frac{n-1}{n} \right) \left(\frac{1}{n-1} \sum_{i=1}^{n-1} \frac{e^{r(X_i)}}{\mu}\right)}  \right]  \right|.
\end{align*}
Our goal is to bound the term
\begin{equation}
    L_n(x) \triangleq \mathbb{E}_{P^{\otimes (n-1)}}\left[\frac{1}{\frac{e^{r(x)}}{n \mu}+ \left(\frac{n-1}{n} \right) \left(\frac{1}{n-1} \sum_{i=1}^{n-1} \frac{e^{r(X_i)}}{\mu}\right)}  \right] .
\end{equation}
Denote $Y_i\triangleq  \frac{e^{r(X_i)}}{\mu}$, and
\begin{equation}
    S_{k} \triangleq \frac{1}{k} \sum_{i=1}^{k} Y_i. 
\end{equation}
Note that $\mathbb{E}[Y_i]=1$, and $S_k$ has finite variance and third moment, denoted by $\sigma^2$ and $\rho$, respectively. For all $x$, we have $S_{n-1} \geq \frac{1}{\mu}$ since $e^{r(x)}/\mu \geq 1/\mu$ for all $x$. This allows us to introduce the unit step function $u(t)$, we can thus rewrite $L_n(x)$ as
\begin{align*}
    L_n(x) &= \left(\frac{n}{n-1}\right) \mathbb{E}\left[\frac{1}{\frac{e^{r(x)}}{(n-1) \mu}+ S_{n-1}}  \right]
\end{align*}
We can now calculate this expectation and bound it with the Berry-Esseen theorem. First, note that $S_{n-1}$ is an average of i.i.d. random variables with:
\begin{itemize}
\item $E[Y_i] = 1$ (by definition)
\item $\text{Var}(Y_i) = \sigma^2$ (finite by assumption)
\item $E[|Y_i - E[Y_i]|^3] = \rho$ (finite by assumption)
\end{itemize}
Then, the Berry-Esseen theorem says
\begin{thm}[Berry-Esseen]
For any $t \in \mathbb{R}$:
\[\left|P(S_{n-1} \leq t) - \Phi\left(\frac{(t-1)\sqrt{n-1}}{\sigma}\right)\right| \leq \frac{C\rho}{\sigma^3\sqrt{n-1}}\]
where $\Phi$ is the standard normal CDF and $C$ is an absolute constant.
\end{thm}

As a consequence of the theorem, we have the upper and lower bounds:
\begin{align*}
P(S_{n-1} \leq t) &\leq \Phi\left(\frac{(t-1)\sqrt{n-1}}{\sigma}\right) + \frac{C\rho}{\sigma^3\sqrt{n-1}} \\
P(S_{n-1} \leq t) &\geq \Phi\left(\frac{(t-1)\sqrt{n-1}}{\sigma}\right) - \frac{C\rho}{\sigma^3\sqrt{n-1}}
\end{align*}

Denoting $c=\frac{e^{r(x)}}{(n-1) \mu}$ and using that $E[X] = \int_0^{\infty} P(X > t) dt$, we have
\begin{align*}
&\mathbb{E}\left[\frac{1}{\frac{e^{r(x)}}{(n-1) \mu}+ S_{n-1}}\right] = \int_0^{\infty} \Pr\left[\frac{1}{c+S_{n-1}} > t \right] dt\\
&= \int_0^{\mu/(\mu c+1)} \Pr\left[S_{n-1} < \frac{1}{t}-c \right] dt \quad \text{(since }S_{n-1} \geq \frac{1}{\mu}\text{ always)}\\
&= \int_{\frac{1}{1-\delta+c}}^{\mu/(\mu c+1)} \Pr\left[S_{n-1} < \frac{1}{t}-c \right]dt + \int_0^{\frac{1}{1-\delta+c}} \Pr\left[S_{n-1} < \frac{1}{t}-c \right]dt\\
&\leq \int_{\frac{1}{1-\delta+c}}^{\mu/(\mu c+1)} \left[\Phi\left(\frac{\sqrt{n-1}}{\sigma}\left(\frac{1}{t} - c - 1\right)\right) + \frac{C\rho}{\sigma^3\sqrt{n-1}}\right] dt + \int_0^{\frac{1}{1-\delta+c}}1 dt\\
&\leq \int_{\frac{1}{1-\delta+c}}^{\mu/(\mu c+1)} \Phi\left(\frac{\sqrt{n-1}}{\sigma}\left(\frac{1}{t} - c - 1\right)\right)dt + \frac{C\rho}{\sigma^3\sqrt{n-1}} \cdot \frac{\mu}{\mu c+1} + \frac{1}{1-\delta+c}
\end{align*}

\begin{align*}
&\mathbb{E}\left[\frac{1}{c+ S_{n-1}}\right] \leq \int_{\frac{1}{1-\delta+c}}^{\mu/(\mu c+1)} \Phi\left(\frac{\sqrt{n-1}}{\sigma}\left(\frac{1}{t} - c - 1\right)\right)dt + \frac{C\rho}{\sigma^3\sqrt{n-1}} \cdot \frac{\mu}{\mu c+1} + \frac{1}{1-\delta+c}
\end{align*}

Now, for the main integral term, when $t > \frac{1}{1-\delta+c}$, we have $\frac{1}{t} - c - 1 < (1-\delta+c) - c - 1 = -\delta$. Therefore $\Phi\left(\frac{\sqrt{n-1}}{\sigma}(\frac{1}{t} - c - 1)\right) \leq \Phi\left(-\delta\frac{\sqrt{n-1}}{\sigma}\right)$. Using the bound for normal CDF: For $z > 0$,
$\Phi(-z) = 1 - \Phi(z) \leq \frac{1}{z\sqrt{2\pi}}e^{-z^2/2}$
Applying this with $z = \delta\frac{\sqrt{n-1}}{\sigma}$:
\begin{align*}
&\int_{\frac{1}{1-\delta+c}}^{\mu/(\mu c+1)} \Phi\left(\frac{\sqrt{n-1}}{\sigma}\left(\frac{1}{t} - c - 1\right)\right)dt\leq \frac{\mu}{\mu c+1} \cdot \frac{\sigma}{\delta\sqrt{2\pi(n-1)}}e^{-\delta^2(n-1)/(2\sigma^2)}
\end{align*}
Therefore:
\begin{align*}
\mathbb{E}\left[\frac{1}{c+ S_{n-1}}\right] &\leq \frac{\mu}{\mu c+1} \cdot \frac{\sigma}{\delta\sqrt{2\pi(n-1)}}e^{-\delta^2(n-1)/(2\sigma^2)} + \frac{C\rho}{\sigma^3\sqrt{n-1}} \cdot \frac{\mu}{\mu c+1} + \frac{1}{1-\delta+c}
\end{align*}

The exponential term decays very fast with n, so it can be absorbed into the error term. Note that $\delta$ appears as a free parameter in our bounds, with the only constraint that $\delta < 1-\frac{1}{\mu}$. Choosing $\delta = \frac{4C'}{\sqrt{n-1}}$ for some constant $C'$ (which is valid for large enough n since we need $\delta < 1-\frac{1}{\mu}$), we get:
\begin{align*}
\mathbb{E}\left[\frac{1}{c+ S_{n-1}}\right] &\leq 1 - \frac{\delta}{2} + \frac{C'}{\sqrt{n-1}} + O\left(\frac{1}{n^{3/2}}\right)\
&= 1 - \frac{2C'}{\sqrt{n-1}} + \frac{C'}{\sqrt{n-1}} + O\left(\frac{1}{n^{3/2}}\right)\
&= 1 - \frac{C'}{\sqrt{n-1}} + O\left(\frac{1}{n^{3/2}}\right)
\end{align*}
This gives us the desired lower bound on $|1 - L_n(x)|$:
\begin{align*}
|1 - L_n(x)| \geq \frac{C'}{\sqrt{n-1}} - O\left(\frac{1}{n^{3/2}}\right)
\end{align*}

Substituting this back into the TV distance:
\begin{align*}
\TV(P^*_\lambda \| P_{n,\lambda}) &= \frac{1}{2}\sum_x \frac{P(x)e^{r(x)}}{\mu}|1 - L_n(x)| \\
&\geq \frac{1}{2}\sum_x \frac{P(x)e^{r(x)}}{\mu}\left(\frac{C'}{\sqrt{n-1}} - O\left(\frac{1}{n^{3/2}}\right)\right) \\
&= \frac{1}{2}\left(\frac{C'}{\sqrt{n-1}} - O\left(\frac{1}{n^{3/2}}\right)\right)\sum_x \frac{P(x)e^{r(x)}}{\mu} \\
&= \frac{1}{2}\left(\frac{C'}{\sqrt{n-1}} - O\left(\frac{1}{n^{3/2}}\right)\right) \cdot 1 \\
&= \frac{C'}{2\sqrt{n-1}} - O\left(\frac{1}{n^{3/2}}\right)
\end{align*}
The last step uses the fact that $\sum_x \frac{P(x)e^{r(x)}}{\mu} = 1$ since $P^*_\lambda$ is a probability distribution. This lower bound on TV distance then implies through Pinsker's inequality that:
\begin{equation}
    D_{\KL}(P^*_\lambda \| P_{n,\lambda}) \geq 2 [\TV(P^*_\lambda \| P_{n,\lambda})]^2 \geq \Omega\left(\frac{1}{n}\right)
\end{equation}

\subsection{Relative Reward Bound.}\label{sec:relative_lower_bound}

While our previous analysis focused on establishing convergence rates in KL  divergence between Soft Best-of-$n$ sampling and the optimal tilted distribution, we now turn to analyzing how quickly the expected rewards converge. Specifically, we establish that the relative difference in expected rewards decays as $O(1/n)$, matching our KL divergence bounds. This provides a complete picture of how Soft Best-of-$n$ sampling approximates both the distribution and reward properties of the optimal solution. 

\ThmRewardBound*

\begin{remark}[Sub-Gaussian rewards via Donsker–Varadhan]
The \(O(n^{-1})\) relative reward gap in Theorem \ref{thm:relative_bound} relies on the reward being bounded.  
If instead the centered reward
\[
\tilde r \;=\; r_\lambda(X)-\mathbb{E}_{P_{n,\lambda}}\!\bigl[r_\lambda(X)\bigr]
\]
is merely $\sigma^2$-sub-Gaussian under \(P_{n,\lambda}\), i.e.
\(
\log \mathbb{E}_{P_{n,\lambda}}\!\bigl[\exp(\alpha\tilde r)\bigr] \le \alpha^{2}\sigma^{2}/2
\) for all \(\alpha\in\mathbb{R}\),
then combining Theorem 1 with the Donsker–Varadhan variational inequality yields
\[
\frac{\mathbb{E}_{P^*_\lambda}[r_\lambda]-\mathbb{E}_{P_{n,\lambda}}[r_\lambda]}
      {\mathbb{E}_{P^*_\lambda}[r_\lambda]}
\;\;\le\;\;
(\lambda+1)\,
\frac{\sinh\!\bigl(\tfrac{1}{2\lambda}\bigr)\,\sigma}
     {\sqrt{8\,n}}
\;=\;O\!\Bigl(n^{-1/2}\Bigr).
\]
Thus one can dispense with the bounded-reward assumption at the cost of a slower
\(1/\sqrt{n}\) rate.
\end{remark}

\subsection{Blockwise sampling}\label{sec:appendix_blockwise_sampling}

Having analyzed Soft Best-of-$n$ sampling for individual tokens, we now extend our analysis to sequences of tokens. In oarticular, we study how the sampling strategy performs when applied to blocks of $m$ tokens at once rather than token-by-token. This analysis reveals fundamental trade-offs between sequence length, sample complexity, and approximation quality that are crucial for practical applications. We begin by introducing a useful Lemma, similar to Lemma \ref{lem:p-char} but for blocks of size $m$. 
\begin{lem}
For a sequence $y^m \in \mathcal{X}^m$, $n \geq 1$, $\lambda > 0$, and a sequence of independent samples $X_1^m, \ldots, X_n^m$ drawn i.i.d. from $P_X^{\otimes m}$, let $Z$ be the random index drawn according to step 2 of Definition \ref{defn:n-tilt-blockwise}. Then the probability that the n-tilted sampling strategy outputs $y^m$ is given by:

\begin{align*}
P_{n,\lambda}(y^m) &= n \cdot P_X^{\otimes m}(y^m) \cdot \mathbb{E}\left[\frac{e^{r(y^m)/\lambda}}{e^{r(y^m)/\lambda} + \sum_{j=2}^n e^{r(X_j^m)/\lambda}}\right] \\
&\geq P_X^{\otimes m}(y^m)e^{r(y^m)/\lambda} \cdot \frac{1}{\frac{1}{n}e^{r(y^m)/\lambda} + \frac{n-1}{n}\left(\mathbb{E}_{P_X^{\otimes m}}[e^{r(X^m)/\lambda}]\right)}
\end{align*}

where
\[\mathbb{E}_{P_X^{\otimes m}}[e^{r(X^m)/\lambda}] = \left(\mathbb{E}_P[e^{r(X)/m\lambda}]\right)^m\]

due to the additive structure of the reward and independence.
\end{lem}

\begin{proof}
Following Definition 2 for sequences, we have:
\begin{align*}
P_{n,\lambda}(y^m) &= \sum_{i=1}^n \text{Pr}(Y^m = y^m, Z = i) \\
&= \sum_{i=1}^n \text{Pr}(X_i^m = y^m, Z = i) \\
&= \sum_{i=1}^n \text{Pr}(X_i^m = y^m) \cdot \text{Pr}(Z = i|X_i^m = y^m)
\end{align*}

By symmetry, we can focus on $i=1$ and multiply by $n$:
\begin{align*}
P_{n,\lambda}(y^m) &= n \cdot P_X^{\otimes m}(y^m) \cdot \text{Pr}(Z = 1|X_1^m = y^m) \\
&= n \cdot P_X^{\otimes m}(y^m) \cdot \mathbb{E}\left[\frac{e^{r(y^m)/\lambda}}{e^{r(y^m)/\lambda} + \sum_{j=2}^n e^{r(X_j^m)/\lambda}}\right]
\end{align*}

Since $f(x) = 1/x$ is convex for $x > 0$, we can apply Jensen's inequality:
\begin{align*}
\mathbb{E}\left[\frac{e^{r(y^m)/\lambda}}{e^{r(y^m)/\lambda} + \sum_{j=2}^n e^{r(X_j^m)/\lambda}}\right] &\geq \frac{e^{r(y^m)/\lambda}}{e^{r(y^m)/\lambda} + (n-1)\mathbb{E}[e^{r(X^m)/\lambda}]} \\
&= \frac{1}{\frac{1}{n} + \frac{n-1}{n}\frac{\mathbb{E}[e^{r(X^m)/\lambda}]}{e^{r(y^m)/\lambda}}}
\end{align*}

For the expectation term, due to the additive structure of the reward and independence:
\begin{align*}
\mathbb{E}[e^{r(X^m)/\lambda}] &= \mathbb{E}\left[\exp\left(\frac{1}{m\lambda}\sum_{i=1}^m r(X_i)\right)\right] = \mathbb{E}\left[\prod_{i=1}^m e^{r(X_i)/m\lambda}\right] = \left(\mathbb{E}_P[e^{r(X)/m\lambda}]\right)^m
\end{align*}

Combining these results yields the stated bound.
\end{proof}

\subsection{Convergence Upper Bound for the Blockwise Case}\label{sec:proof_upperbound_blockwise}

We begin by restating Theorem \ref{thm:KL-additive-reward}.

\ThmKLAdditive*

\begin{proof}
We begin by examining the blockwise n-tilted sampling procedure on sequences of length $m$. In this setting, we sample $n$ sequences $X_1^m,\ldots,X_n^m$ independently from $P_X^{\otimes m}$ and select one according to the blockwise $n$-tilted sampling from Def. \ref{defn:n-tilt-blockwise}. Let us start with the KL divergence expression:
\begin{align*}
\KL(P^*_\lambda \|P_{n,\lambda}) &= \sum_{y^m} P^*_\lambda(y^m) \log \frac{P^*_\lambda(y^m)}{P_{n,\lambda}(y^m)} \\
&= \sum_{y^m} \frac{P^{\otimes m}_X(y^m)e^{r(y^m)/\lambda}}{\mathbb{E}_{P^{\otimes m}_X}[e^{r(X^m)/\lambda}]} \log \frac{P^{\otimes m}_X(y^m)e^{r(y^m)/\lambda}}{\mathbb{E}_{P^{\otimes m}_X}[e^{r(X^m)/\lambda}] \cdot P_{n,\lambda}(y^m)}
\end{align*}

Using the lower bound from Lemma 2:
\[
P_{n,\lambda}(y^m) \geq P^{\otimes m}_X(y^m)e^{r(y^m)/\lambda} \cdot \frac{1}{\frac{1}{n}e^{r(y^m)/\lambda} + \frac{n-1}{n}\mathbb{E}_{P^{\otimes m}_X}[e^{r(X^m)/\lambda}]}
\]

Substituting this bound:
\begin{align*}
\KL(P^*_\lambda \|P_{n,\lambda}) &\leq \sum_{y^m} \frac{P^{\otimes m}_X(y^m)e^{r(y^m)/\lambda}}{\mathbb{E}_{P^{\otimes m}_X}[e^{r(X^m)/\lambda}]} \log \frac{\frac{1}{n}e^{r(y^m)/\lambda} + \frac{n-1}{n}\mathbb{E}_{P^{\otimes m}_X}[e^{r(X^m)/\lambda}]}{\mathbb{E}_{P^{\otimes m}_X}[e^{r(X^m)/\lambda}]} \\
&= \sum_{y^m} \frac{P^{\otimes m}_X(y^m)e^{r(y^m)/\lambda}}{\mathbb{E}_{P^{\otimes m}_X}[e^{r(X^m)/\lambda}]} \log \left(1 + \frac{1}{n}\left(\frac{e^{r(y^m)/\lambda}}{\mathbb{E}_{P^{\otimes m}_X}[e^{r(X^m)/\lambda}]} - 1\right)\right)
\end{align*}

Since $\log$ is concave, we can apply Jensen's inequality:
\begin{align*}
&\leq \log\left(1 + \frac{1}{n}\sum_{y^m} \frac{P^{\otimes m}_X(y^m)e^{r(y^m)/\lambda}}{\mathbb{E}_{P^{\otimes m}_X}[e^{r(X^m)/\lambda}]} \left(\frac{e^{r(y^m)/\lambda}}{\mathbb{E}_{P^{\otimes m}_X}[e^{r(X^m)/\lambda}]} - 1\right)\right) \\
&= \log\left(1 + \frac{1}{n}\left(\frac{\mathbb{E}_{P^{\otimes m}_X}[e^{2r(X^m)/\lambda}]}{(\mathbb{E}_{P^{\otimes m}_X}[e^{r(X^m)/\lambda}])^2} - 1\right)\right)
\end{align*}

For the additive reward model where $r(x^m) = \frac{1}{m}\sum_{i=1}^m r(x_i)$, using the independence of samples:
\begin{align*}
\frac{\mathbb{E}_{P^{\otimes m}_X}[e^{2r(X^m)/\lambda}]}{(\mathbb{E}_{P^{\otimes m}_X}[e^{r(X^m)/\lambda}])^2} - 1 
&= \frac{\mathbb{E}_{P^{\otimes m}_X}[e^{\frac{2}{m\lambda}\sum_{i=1}^m r(X_i)}]}{(\mathbb{E}_{P^{\otimes m}_X}[e^{\frac{1}{m\lambda}\sum_{i=1}^m r(X_i)}])^2} - 1 \\
&= \left(\frac{\mathbb{E}_P[e^{2r(X)/m\lambda}]}{(\mathbb{E}_P[e^{r(X)/m\lambda}])^2}\right)^m - 1 \\
&= \text{CV}(e^{r(X)/m\lambda})^{2m}
\end{align*}

where $\text{CV}(e^{r(X)/m\lambda})^2 = \frac{\mathbb{E}_P[e^{2r(X)/m\lambda}]}{(\mathbb{E}_P[e^{r(X)/m\lambda}])^2} - 1$ is the squared coefficient of variation.

Therefore:
\[
\KL(P^*_\lambda \|P_{n,\lambda}) \leq \log\left(1 + \frac{1}{n}\text{CV}(e^{r(X)/m\lambda})^{2m}\right)
\]
\end{proof}

\subsection{Proof of Corollary \ref{cor:lambda_additivereward}}\label{sec:corollary_lambda_blockwise}

We begin by restating Corollary \ref{cor:lambda_additivereward}.

\CorLambdaAdditive*

\begin{proof}
From Theorem 4, we have:
\[
\KL(P^*_\lambda \|P_{n,\lambda}) \leq \log\left(1 + \frac{1}{n}\text{CV}(e^{r(X)/m\lambda})^{2m}\right)
\]

When $0 \leq r(x) \leq 1$, we can bound the coefficient of variation using properties of exponential functions and the Bhatia-Davis inequality:
\begin{align*}
\text{CV}(e^{r(X)/m\lambda})^2 &= \frac{\mathbb{E}_P[e^{2r(X)/m\lambda}]}{(\mathbb{E}_P[e^{r(X)/m\lambda}])^2} - 1 \\
&\leq \frac{e^{2/m\lambda} - 1}{4} 
\end{align*}

In our case, let $Y = e^{r(X)/m\lambda}$. Since $0 \leq r(x) \leq 1$, we have:
\begin{align*}
m &= \min_x e^{r(x)/m\lambda} = e^0 = 1 \\
M &= \max_x e^{r(x)/m\lambda} = e^{1/m\lambda}
\end{align*}

Let $\mu = \mathbb{E}[Y]$. By Bhatia-Davis inequality:
\[
\text{Var}(Y) \leq (e^{1/m\lambda} - \mu)(\mu - 1)
\]

The right-hand side is quadratic in $\mu$, and its maximum occurs at $\mu = \frac{e^{1/m\lambda} + 1}{2}$. At this value:
\begin{align*}
\text{Var}(Y) &\leq \left(e^{1/m\lambda} - \frac{e^{1/m\lambda} + 1}{2}\right)\left(\frac{e^{1/m\lambda} + 1}{2} - 1\right) \\
&= \left(\frac{e^{1/m\lambda} - 1}{2}\right)\left(\frac{e^{1/m\lambda} - 1}{2}\right) \\
&= \frac{(e^{1/m\lambda} - 1)^2}{4}
\end{align*}

Now, recall that the coefficient of variation is:
\[
\text{CV}(Y)^2 = \frac{\mathbb{E}[Y^2]}{(\mathbb{E}[Y])^2} - 1 = \frac{\text{Var}(Y)}{(\mathbb{E}[Y])^2} = \frac{\mathbb{E}[e^{2r(X)/m\lambda}]}{(\mathbb{E}[e^{r(X)/m\lambda}])^2} - 1
\]

Therefore:
\begin{align*}
\text{CV}(e^{r(X)/m\lambda})^2 = \frac{\text{Var}(Y)}{(\mathbb{E}[Y])^2} \leq \frac{(e^{1/m\lambda} - 1)^2}{4(\mathbb{E}[Y])^2} \leq \frac{e^{2/m\lambda} - 1}{4}
\end{align*}

where the last inequality follows from $\mathbb{E}[Y] \geq 1$ (since $Y \geq 1$ almost surely).

Therefore:
\begin{align*}
\KL(P^*_\lambda \|P_{n,\lambda}) &\leq \log\left(1 + \frac{1}{n}\left(\frac{e^{2/m\lambda} - 1}{4}\right)^m\right) \leq \epsilon
\end{align*}

For this inequality to hold, it suffices to have:
\begin{align*}
1 + \frac{1}{n}\left(\frac{e^{2/m\lambda} - 1}{4}\right)^m &\leq e^\epsilon \\
\frac{1}{n}\left(\frac{e^{2/m\lambda} - 1}{4}\right)^m &\leq e^\epsilon - 1 \\
\left(\frac{e^{2/m\lambda} - 1}{4}\right)^m &\leq n(e^\epsilon - 1) \\
\frac{e^{2/m\lambda} - 1}{4} &\leq (n(e^\epsilon - 1))^{1/m} \\
e^{2/m\lambda} &\leq 1 + 4(n(e^\epsilon - 1))^{1/m} \\
\frac{2}{m\lambda} &\leq \log(1 + 4(n(e^\epsilon - 1))^{1/m})
\end{align*}

Solving for $\lambda$:
\[
\lambda \geq \frac{2}{m\log(1 + 4(n(e^\epsilon - 1))^{1/m})}
\]
\end{proof}

\end{document}